\documentclass[conference]{IEEEtran}
\ifCLASSINFOpdf
\else
\fi
\usepackage{amsthm}
\usepackage{amsmath}
\usepackage{graphicx}
\usepackage{color}
\usepackage{algorithm}
\usepackage{algorithmic}
\usepackage{diagbox}
\usepackage{multirow}
\usepackage{float} 
\usepackage{subfigure}
\usepackage{amssymb}
\usepackage{multirow}

\newtheorem{definition}{Definition}
\newtheorem{theorem}{Theorem}
\newtheorem{lemma}{Lemma}
\newtheorem{corollary}{Corollary}

\newenvironment{restateTheorem}[1]
    {   \renewcommand{\thetheorem}{\ref{#1}}%
        \addtocounter{theorem}{-1}%
        \begin{theorem}}
    {\end{theorem}}

\author{\IEEEauthorblockN{Zihui Liang}
\IEEEauthorblockA{University of Electronic Science\\ and Technology of China\\
Email: zihuiliang.tcs@gmail.com}
\and
\IEEEauthorblockN{Bakh Khoussainov}
\IEEEauthorblockA{University of Electronic Science\\ and Technology of China\\
Email: bmk@uestc.edu.cn}
\and
\IEEEauthorblockN{Toru Takisaka}
\IEEEauthorblockA{University of Electronic Science\\ and Technology of China\\
Email: t.takisaka2009.a2n7@gmail.com}
\and
\IEEEauthorblockN{Mingyu Xiao}
\IEEEauthorblockA{University of Electronic Science\\ and Technology of China\\
Email: myxiao@uestc.edu.cn}}

\begin{document}
%
\title{Connectivity in the presence of an opponent}



%


\maketitle

\begin{abstract}
The paper introduces two player connectivity games played on finite bipartite graphs. Algorithms that solve these connectivity games can be used as subroutines for solving M\"uller games. M\"uller games constitute a well established class of games in model checking and verification.  In connectivity games, the objective of one of the players is to visit every node of
the game graph infinitely often. The first contribution of this paper is our proof that solving connectivity games can be reduced to the incremental strongly connected component maintenance (ISCCM) problem, an important problem in graph algorithms and data structures.  The second contribution is that we non-trivially adapt two known algorithms for the ISCCM problem to provide two efficient algorithms that solve the connectivity games problem. Finally, based on the techniques developed, we 
recast Horn's polynomial time algorithm that solves  explicitly given M\"uller games and provide an alternative proof of its correctness. Our algorithms are more efficient than that of Horn's algorithm. Our solution for connectivity games is used as a subroutine in the algorithm. 
\end{abstract}

\section{Introduction}

\subsection{M\"uller games given explicitly} 
In the area of logic, model checking,  and verification of reactive systems, studying games played on graphs  is a key research topic \cite{chatterjee2014efficient} \cite{gradel2002automata}. This is mostly motivated 
through modelling reactive systems and reductions of model checking problems to games on graphs. 
Understanding the algorithmic content of determinacy results is also at the core of this research.   
M\"uller games constitute a well-established class of games for verification. Recall that a {\em M\"uller game} $\mathcal G$ is a tuple $(V_0,V_1,E,\Omega)$, where
\begin{itemize}
    \item The tuple $G=(V_0\cup V_1, E)$ is a finite directed bipartite graph so that $V_0$ and $V_1$ partition the set 
    $V=V_0\cup V_1$. Usually $G$ is called the arena of $\mathcal G$.  
    \item The set $E\subseteq (V_0\times V_1) \cup (V_1\times V_0)$ of edges. 
    \item $V_0$ and $V_1$ are sets from which player 0 and player 1, respectively, move. Positions in $V_{\epsilon}$
     are called player $\epsilon$ positions, $\epsilon \in\{0,1\}$. 
    \item $\Omega\subseteq 2^V$ is a collection of winning sets. 
\end{itemize} 
 Say that 
the  {\em game} $\mathcal G=(V_0,V_1,E,\Omega)$ is {\em explicitly given} if $V$, $E$, and all sets in $\Omega$ are fully presented as input. The (input) size of explicitly given M\"uller game is thus bounded by $|V|+|E|+ 2^{|V|}\cdot |V|$. Finally, the {\em game graph} of the M\"uller game $\mathcal G$  is the underlying bipartite graph $G=(V_0 \cup V_1, E)$. 

\smallskip

For each $v\in V$, consider the set $E(v)=\{u \mid E(v,u)\}$ of successors of $v$. Let $X\subseteq V$. Call the set $E(X)=\bigcup_{v\in X} E(v)$ the successor of $X$.  
Similarly, for a $v \in V$, the predecessor of $v$ is the set $E^{-1}(v)=\{u\mid (u,v)\in E\}$. Call the set $E^{-1}(X)=\bigcup_{v\in X} E^{-1}(v)$ the predecessor of $X$.

Let $\mathcal G=(V_0,V_1,E,\Omega)$ be a M\"uller game. The players play the game by moving a given token along the edges of the graph. The token is initially placed on a node $v_0 \in V$. The play proceeds in rounds. At any round of the play, if the token is placed on a player $\sigma$'s node $v$, then player $\sigma$ chooses $u\in E(v)$, moves the token to $u$ and the play continues on to the next round. Formally, a play (starting from $v_0$) is a sequence $\rho=v_0,v_1,\ldots$ such that $v_{i+1}\in E(v_i)$ for all $i\in \mathbb{N}$. If a play reaches a position $v$ such that $E(v)=\emptyset$, then player 1 wins the play. For an infinite play $\rho$,  set $\mathsf{Inf}(\rho)=\{v\in V \mid \exists^{\omega} i (v_i=v)\}$. We say player 0 wins the play $\rho$ if $\mathsf{Inf}(\rho)\in \Omega$; otherwise, player 1 wins the play. 

A strategy for player $\sigma$ is a function that takes as inputs initial segments of plays $v_0,v_1,\ldots, v_k$ where $v_k\in V_\sigma$ and output some $v_{k+1}\in E(v_k)$. A strategy for player $\sigma$ is winning from $v_0$ if, assuming player $\sigma$ follows the strategy, all plays starting from $v_0$ generated by the players are winning for player $\sigma$. The game $\mathcal{G}$ is determined if one of the players has a winning strategy.  M\"uller games are Borel games, and hence, by the result of Martin \cite{martin1975borel}, they are determined. Since M\"uller games are determined we can partition the set $V$ onto two sets $Win_0$ and $Win_1$, where $v\in W_{\epsilon}$ iff player $\epsilon$ wins the game starting at $v$, $\epsilon \in \{0,1\}$.
To solve a given M\"uller game $\mathcal G=(V_0, V_1, E, \Omega)$ means to find the sets $Win_0$ and $Win_1$.
There are several known algorithms that solve M\"uller games. These algorithms provide the basis for analysis and synthesis of M\"uller games. In particular, these algorithms extract finite state winning strategies for the players  \cite{dziembowski1997much,horn2008explicit,hunter2005complexity,mcnaughton1993infinite,nerode1996mcnaughton,zielonka1998infinite}. \ We stress that the algorithms that solve M\"uller games depend on the presentations of the games. 
The problem of solving M\"uller games  is typically in PSPACE for many reasonable representations 
\cite{mcnaughton1993infinite,nerode1996mcnaughton}.  However, if the winning condition is represented as a Zielonka tree \cite{zielonka1998infinite} or as the well-known parity condition, then solving the games turns into a $NP\cap co$-$NP$ problem. P. Hunter and A. Dawar \cite{hunter2005complexity} investigate five 
other representations: win-set, Muller, Zielonka DAGs, Emerson-Lei, and explicit Muller.  They show
that the problem of the winner is PSPACE-hard for the first four representations. F. Horn \cite{horn2008explicit} provides a polynomial time algorithm that solves explicit M\"uller games. However, his proof of corrextness has some flaws. 
So, we provide an alternative correctness proof of the algorithm. 
Designing new algorithms, improving and analysing the state of the art techniques in this area is a key research direction. This paper contributes to this. 


\subsection{Connectivity games} 

One motivation for defining connectivity games comes from solving M\"uller games. 
Many algorithms that solve M\"uller games or its variants are recursive. Given a M\"uller game $\mathcal G$, one constructs a set of  smaller M\"uller games. The solution of the 
games $\mathcal G'$ from this set is then used to solve $\mathcal G$. Through an iteration process, these reductions produce sequences of the form $\mathcal G_1, \mathcal G_2, \ldots, \mathcal G_r$, where $\mathcal G_{i+1}= \mathcal G_i'$ such that $\mathcal G_r=\mathcal G_{r+1}$. The key point is that solving the game $\mathcal G_r$ at the base of this iteration boils down to investigating connectivity of the graph $G_r$ in the game-theoretic setting. Namely, to win the game $\mathcal G_r$, one of the players  has to visit all the nodes of $\mathcal G_r$ infinitely often.  This observation calls for deeper and refined analysis of those M\"uller games $\mathcal G=(V_0, V_1, E, \Omega)$ where the objective of player 0 is to visit all the nodes of the underlying graph $G$, that is, $\Omega=\{V\}$. We single out these games:

\begin{definition} \label{Defn:Connectivity}
A M\"uller game $\mathcal G=(V_0, V_1, E, \Omega)$ is called a {\bf connectivity game} if $\Omega$ is a singleton that consists of $V$. 
\end{definition}

The second motivation to investigate the connectivity games comes from the concept of connectivity itself.  The notion of (vertex) connectivity is  fundamental in graph theory and its applications. There is a large amount of work ranging from complexity theoretic issues to designing efficient data structures and algorithms that aim to analyse connectivity in graphs. Connectivity in graphs and graph like structures is well-studied in almost all areas of computer science in various settings and motivations. 
 For undirected graphs, connectivity of a graph $G$ is defined through existence of paths between all 
vertices of $G$. 
For directed graphs $G$ connectivity is defined through strong connectivity. The digraph $G$ is strongly connected if for any two vertices $x$ and $y$ there exist paths from $x$ to $y$ and from $y$ to $x$. 
One can extend these notions of (vertex) connectivity into a game-theoretic setting as follows. There are two players: player 0 and player 1.  A token is placed on a vertex $v_0$ of a bipartite graph $G=(V_0\cup V_1, E)$. 
Player 0 starts the play by moving the token along an outgoing edge $(v_0, v_1)$. Player 1 responds by moving the token along an outgoing edge  from the vertex $v_1$, say $(v_1, v_2)$. This continues on and the players produce a path $v_0, v_1, \ldots, v_k$ called a play starting at $v_0$. Say that player 0 wins the play $v_0, v_1, \ldots, v_k$ if the play visits every node in $V$. 
Call thus described game {\em forced-connectivity game}. 
If player 0 has a winning strategy, then we say that the player wins the game starting at $v_0$. Winning this forced-connectivity game from $v_0$ does not always guarantee that the player wins the game starting at any other vertex. Therefore we can define game-theoretic connectivity as follows. \ A directed bipartite graph $G$ is {\em forced-connected} if player 0 wins the forced-connectivity game in $G$ starting at any vertex of $G$. Thus, finding out if $G$ is a forced-connected is equivalent to solving the connectivity games as in Definition \ref{Defn:Connectivity}. 


\begin{definition}\label{Defn:connected}
Let $\mathcal G=(V_0, V_1, E, \Omega)$ be a connectivity game. Call the bipartite graph $G=(V_0, V_1, E)$ {\bf forced-connected} if player 0 wins the game $\mathcal G$. 
\end{definition}

The third motivation is related to generalised B\"uchi winning condition.  The generalised B\"uchi winning condition is given by subsets   $F_1$, $\ldots$, $F_k$ of the game graph $G$. Player 0 wins a play if the play meets each of these winning sets   $F_1$, $\ldots$, $F_k$ infinitely often. Our connectivity games winning condition can be viewed as  a specific generalised B\"uchi winning condition where the accepting sets are all singletons.

\subsection{Our contributions}

The focus of this paper is two-fold. On the one hand, we study connectivity games and provide the state-of-the-art algorithms for solving them.  H. Bodlaender, M. Dinneen, and B. Khoussainov \cite{bodlaender2002relaxed} call connectivity games {\em update games}. Their motivation comes from modelling the scenario where messages should be passed to all the nodes of the network in the presence of adversary. On the other hand, using  the connectivity game solution process as a subroutine, we 
recast Horn's polynomial time algorithm that solves 
explicitly given M\"uller games and provide an alternative proof of its correctness. We now detail these below in describing our  contributions.
 
\smallskip
\noindent 
{\bf 1.}  Our first contribution is that given a connectivity game $\mathcal G$, we construct a sequence of directed graphs $\mathcal G_0, \mathcal G_1, \ldots, \mathcal G_s$ such that player 0 wins $\mathcal G$ if and only if $\mathcal G_s$ is strongly connected [See Theorem \ref{Thm:forced connected to SCC}]. 
Due to this result, we reduce solving connectivity game problem to the incremental strongly connected component maintenance (ISCCM) problem, one of the key problems in graph algorithms and data structure analysis \cite{bender2009new} \cite{haeupler2012incremental}. 

\smallskip
\noindent 
{\bf 2.}  A standard brute-force algorithm that solves the connectivity game $\mathcal G$ runs in time $\mathbf{O}(|V|^2 (|V|+|E|))$.  H. Bodlaender, M. Dinneen, and B. Khoussainov in \cite{bodlaender2002relaxed}  \cite{dinneen2000update} 
provided algorithms that solve the connectivity games in $\mathbf{O}(|V||E|)$. 
Due to Theorem \ref{Thm:forced connected to SCC}, we solve the connectivity game problem by adapting two known algorithms that solve the ISCCM  problem.  The first algorithm is by Haeupler et al. \cite{haeupler2012incremental} who designed the {\em soft-threshold search algorithm}  that  handles sparse graphs. Their algorithm runs in time $\mathbf{O}(\sqrt{m}m)$, where $m$ is the number of edges. The second is the solution by Bender et al. \cite{bender2009new,bender2011new}. Their algorithm is best suited for the class of dense graphs and runs in time of $\mathbf{O}(n^{2}\log n)$, where $n$ is the number of vertices. By adapting these algorithms, we design new algorithms to solve  the connectivity games. The first algorithm, given a connectivity game $\mathcal G$, runs in time  $\mathbf{O}((\sqrt{|V_1|}+1) |E|+ |V_1|^2)$ [See Theorem \ref{Thm:|V_0|}]. The first feature of this algorithm is that the algorithm solves the problem in linear time in $|V_0|$
if  $|V_1|$ is considered as a parameter. The parameter constant in this case is $|V_1|^{3/2}$. The second feature is that the algorithm runs in linear time if the underlying game graph is sparse.  \ Our second algorithm solves the connectivity game in time  $\mathbf{O}((|V_1|+|V_0|)\cdot |V_0|\log |V_0|)$ [See Theorem \ref{Thm:|V_1|}].   In contrast to the previous algorithm, this algorithm  solves the connectivity game problem in linear time in $|V_1|$  if  $|V_0|$ is considered as a parameter. The parameterised constant is $|V_0|\log{|V_0|}$. Furthermore, 
 the second algorithm is more efficient than the first one on dense graphs. These two algorithms outperform 
the standard bound $O(|V||E|)$, mentioned above, for solving the connectivity games. As a framework, this is  similar to the work of K. Chatterjee and M. Henzinger \cite{chatterjee2014efficient} who   improved the standard $O(|V|\cdot |E|)$ time bound for solving  B\"uchi games to $O(|V|^2)$ bound  through analysis of  maximal end-component decomposition algorithms in graph theory.

\smallskip
\noindent 
{\bf 3.} In \cite{horn2008explicit} Horn provided a polynomial time algorithm that solves explicitly given M\"uller games. In his algorithm, Horn uses the standard procedure of solving connectivity games as a subroutine. Directly using our algorithms above, as a subroutine to Horn's algorithm, we obviously improve Horn's algorithm in an order of magnitude. Horn's proof of correctness uses three lemmas (see Lemmas 5, 6, and 7 in \cite{horn2008explicit}). Lemmas 6 and 7  are inter-dependent. 
The proof of Lemma 5 is correct. However, Lemma 7 
contains, in our view, unrecoverable mistake. This is presented in Section \ref{SS:Horns failure}. 
We provide our independent and alternative proof of correctness of Horn's algorithm. To the best of our knowledge, this is the first work that correctly and fully recasts Horn's polynomial time algorithm with the efficient sub-routine for solving the connectivity games. Furthermore, in terms of running time, our algorithms perform better than that of Horn's algorithm [See Theorem \ref{Thm:Explicit Muller game result 1} and Theorem \ref{Thm:Explicit Muller game result 2}]. For instance, one of our algorithms decreases the degree of $|\Omega|$ from $|\Omega|^3$ in Horn's algorithm to $|\Omega|^2$ [See Theorem  \ref{Thm:Explicit Muller game result 2}]. Since $|\Omega|$ is bounded by $2^{|V|}$, the improvement is significant.

\section{
A characterization theorem }

A M\"uller game $\mathcal G=(V_0,V_1,E,\Omega)$ is a connectivity game if $\Omega =\{V\}$. In this section we focus on connectivity games $\mathcal G$. 
In the study of M\"uller games, often it is required that  for each $v$ the successor set $E(v)=\{u\mid (v,u)\in E\}$ is not empty.  We do not put this condition as 
it will be convenient for our analysis of connectivity games to consider cases when $E(v)=\emptyset$. 
Recall that a strongly connected component of a directed graph is a maximal set $X$ such that there exists a path between any two vertices of $X$. Denote the collection of all strongly connected components of the game graph  $G$ of game $\mathcal G$ by $SCC(\mathcal G)$.  For all distinct components $X,Y \in SCC(\mathcal{G})$, we have $X\cap Y=\emptyset$ and $\bigcup_{X\in SCC(\mathcal{G})} X=V$.

\begin{definition}
Let $\mathcal{G}$ be a connectivity game. Consider two sets $U\subseteq V_1$ and $S\subseteq V$.
Define 
        $$
        Force(U, S)=\{v\mid v\in (E^{-1}(S)\setminus S)\cap U\text{ and } E(v)\subseteq S\}.
        $$ 
        \end{definition}
        \begin{definition}
  We say that a set $X\subseteq V$ in game $\mathcal G$  is {\bf forced trap (FT)}  if  either $|X|=1$ or 
        if $|X|>1$ then $E(X\cap V_1)\subseteq X$ and any two vertices in $X$ are strongly connected. 
  
\end{definition}


\begin{lemma} \label{L:FT partition forced-connected game}
    Let $\mathcal C=\{C_1,C_2,\ldots,C_k\}$, where $k>1$, be a collection of FTs that partition the game graph $G$. If $\mathcal{G}$ is forced-connected then
     for every $X\in \mathcal C$ there is a distinct $Y\in \mathcal C$ such that 
  either $Y$ is a singleton consisting of a player 1's node and $E(Y)\subseteq X$, or
         $Y$ has player 0's node $y$ with $E(y)\cap X\ne \emptyset$.
    \end{lemma}

\begin{proof}
    Assume that there is an $X\in \mathcal C$ that doesn't satisfy the lemma. Let $Y\in \mathcal C$. 
    If $Y$  is not a singleton, 
    no move exists from $Y$ into $X$ by any of the players.  If $Y=\{y\}$ is a singleton and $y\in V_1$, then the player can always move outside of $X$ from $y$. If $Y=\{y\}$ is a singleton and $y\in V_0$, then $E(y)\cap X=\emptyset$. Hence, player 1 can always avoid nodes in $X$. This contradicts the assumption that $\mathcal G$ is forced-connected. 
\end{proof}
We now define the sequence $\{\mathcal{G}_k\}_{k\ge 0}$ of graphs. We will call each $\mathcal G_k$  
 the $k^{th}$-derivative of $\mathcal{G}$.  We will also view each $\mathcal G_k$ as a connectivity game. Our construction is the following. 
 \begin{itemize}
    \item Initially, for $k=0$, set $F_0=\emptyset$, $U_0=V_1$ and $\mathcal{G}_0=(V_0,V_1,E_0)$, where $E_0$ consists of all outgoing edges in $E$ of player 0.  
    
    \item For $k>0$, consider the set 
    $F_{k}=\bigcup_{S\in SCC(\mathcal{G}_{k-1})} Force(U_{i-1},S)$, and define
    $U_{k}=U_{k-1}\setminus F_{k}$, $\mathcal{G}_k=(V_0,V_1,E_k)$, where $E_k=E_{k-1}\cup \{(v,u)\mid v\in F_{k} \text{ and } (v,u)\in E\}$. 
\end{itemize}
 Note that the SCCs of $\mathcal G_0$ are all  singletons.  For $k=1$ we have the following. The set $F_1$ consists of all player 1 nodes of out-degree 1. The set $E_1$ contains $E_0$ and all outgoing edges from the set $F_1$. We note that each SCC of $\mathcal G_1$ is also a FT in $\mathcal G_1$.  Therefore each SCC $X$ in $\mathcal G_1$ is also a maximal FT.  
 Observe that each $F_k\subseteq U_{k-1}$ consists of  player 1's nodes $v$ such that all moves 
 of player 1 from $v$ go  to the same SCC in $\mathcal{G}_{k-1}$. Moreover, $U_k$ is the set of player 1's nodes whose outgoing edges aren't in $E_{k}$.  Now we list some properties of the sequence  $\{\mathcal{G}_k\}_{k\ge 0}$. A verification of these properties follows from the construction:
 
\begin{itemize}
   \item For every player 1's node $v$ and $k>0$, the outgoing edges of $v$ are in $E_k\setminus E_{k-1}$ iff 
  all the outgoing edges of $v$ point to the same SCC in $\mathcal{G}_{k-1}$ and in $\mathcal G_{k-1}$ the out-degree of $v$ is $0$.
  
     \item For each $k\geq 0$, every SCC in $\mathcal{G}_k$ is a FT in $\mathcal{G}_k$.
  
  \item For all $k\geq 0$ we have $F_{k+1}\subseteq U_{k}\subseteq U_{k-1} \subseteq \ldots \subseteq U_0=V_1$.
  
  \item For all $k_1\neq k_2$ we have $F_{k_1}\cap F_{k_2}=\emptyset$. 
  
      \item If $F_k=\emptyset$ with $k>0$ then for all $i\ge k$, $\mathcal{G}_{i}=\mathcal{G}_{k-1}$. 
      We call the minimal such $k$ the {\bf stabilization point} and denote it by $s$.  Note that $s\leq |V_1|$. 

   
   \item If for all $X\in SCC(\mathcal{G}_k)$, either $|X|>1$ or $X$  is a singleton consisting of player 0's node only then $\mathcal{G}_k=\mathcal{G}$.

    \item For each $k\geq 0$ and player 1's node $v$, if $v$ is in a nontrivial SCC in $\mathcal G_k$ then all $v$'s outgoing edges from $v$ are in $E_k$.
 
\end{itemize}

\begin{lemma} \label{L:forced connected SCC G_k}
   If $\mathcal{G}$ is forced-connected and $|SCC(\mathcal{G}_{k})|>1$,  
   then $\mathcal{G}_{k}\ne \mathcal{G}_{k+1}$. 
\end{lemma}

\begin{proof}
    Let $SCC(\mathcal{G}_{k})=\{C_1,C_2,\ldots,C_m\}$ with $m>1$. We noted above that each SCC of $\mathcal G_k$ is a FT.  Since $\{C_i\}_{i\le m}$ is a partition of $\mathcal{G}_k$, by Lemma \ref{L:FT partition forced-connected game}, there exists a sequence $P=\{P_1,P_2,\ldots,P_t\}$ of SCCs from $C$ such that     
    \begin{itemize}
        \item If $P_i$ consists of a player 1's node then $E(P_i)\subseteq P_{i \text{ mod } t + 1}$, otherwise
        \item $E(P_i)\cap P_{i \text{ mod }n + 1}\ne \emptyset $.
    \end{itemize}

Let $P_i\in P$. If $P_i$ consists of a player 0's node or $|P_i|>0$, then all outgoing edges from $P_i$ are in $E_k$. 
If $P_i$ consists of a player 1 node, say $P_i=\{v\}$, then all outgoing edges from $v$ are in $E_{k+1}$. This implies that 
$P_1 \cup P_2 \cup \ldots \cup P_t$ is a subset of a strongly connected component of the graph  $\mathcal G_{k+1}$. We conclude that $SCC(\mathcal{G}_k)\ne SCC(\mathcal{G}_{k+1})$ and $\mathcal{G}_k\ne \mathcal{G}_{k+1}$.
\end{proof}

Given an connectivity game $\mathcal G$, we now construct the sequence of forests $\{\Gamma_k (\mathcal G)\}_{k\geq 0}$ by induction. The idea is to represent the interactions of the SCCs of the graphs $\mathcal G_k$ with SCCs of the  $\mathcal G_{k-1}$,  for $k=1,2, \ldots$.   The sequence of forests $\Gamma_k(\mathcal{G})=(N_k, Son_k)$, $k=0,1,\ldots$,  is defined as follows: 

    \begin{itemize}
        \item For $k=0$, set $\Gamma_{0}(\mathcal{G})=(N_{0}, Son_{0})$, where $N_{0}=\{\{v\}\mid v\in V\}$ and $Son_{0}(\{v\})=\emptyset$ for all $v\in V$. 
        \item For $k>0$, let $C=SCC(\mathcal{G}_{k})\setminus N_{k-1}$ be the set of new SCCs in $\mathcal{G}_{k}$. 
        Define the forest $\Gamma_k(\mathcal{G})=(N_k,Son_k)$, where
        \begin{enumerate}
         \item $N_k=N_{k-1} \ \cup \ C$, and 
         \item $Son_{k}=Son_{k-1}\cup \{(X,Y)\mid X\in C , Y\in SCC(\mathcal{G}_{k-1}) \text{ and } Y\subset X\}$.
         \end{enumerate}
         Thus the new SCCs $X$ that belong to $C$ have become the roots of the trees in the forest  $\Gamma_k(\mathcal{G})$. 
         The children of $X$ are now SCCs in $\mathcal G_{k-1}$ that are contained in $X$. 
    \end{itemize}
Note that if $s$ is the stabilization point of the sequence $\{\mathcal G_k\}_{k\geq 0}$, then  for all $k\geq s$ we have  $\Gamma_k(\mathcal{G})= \Gamma_s(\mathcal{G})$. Therefore, we set $\Gamma (\mathcal G)=\Gamma_s (\mathcal G)$. Thus, for the forest $\Gamma (\mathcal G)$ we have $N=N_s$ and $Son=Son_s$.  The following properties of the forest $\Gamma (\mathcal G)$ can easily be verified: 

\begin{itemize}
    \item For all nodes $X\in N$, $X$'s sons partition $X=\bigcup_{Y\in Son(X)}Y$.
    \item For all nodes $X\in N$ with $|X|>1$, $E(X\cap V_1)\subseteq X$.
    \item The roots of $\Gamma(\mathcal{G})$ are strongly connected components of $\mathcal{G}_{s}$.
\end{itemize}

\begin{lemma} \label{L:derivative node connectivity games}
    Consider $\Gamma(\mathcal{G})=(N, Son)$. Let $X \in N$ be such that $|X|>1$. 
    Then the sub-game $\mathcal{G}(X)$ of the game $\mathcal{G}$ played in $X$ is forced-connected.
\end{lemma}

\begin{proof}
    Let $Son(X)=\{C_1,C_2,\ldots, C_m\}$. Note that $m>1$. 
    Since $X$ is a SCC, for every $Y,Z\in Son(X)$, there exists a sequence $P=P_1,P_2,\ldots,P_t$ 
    such that for $i=1,2,\ldots,t$, we have $P_i\in Son(X)$, $P_1=Y$ and $P_t=Z$. In addition, the sequence $P$ satisfies the following properties:
    \begin{itemize}
        \item If $P_i$ consists of a player 1's node, then $E(P_i)\subseteq P_{i+1}$.
        \item If $P_i$ consists of a player 0's node, then player 0 can move $P_i$ to $P_{i+1}$.
        \item If $|P_i|>1$ then $E(P_i) \cap P_{i+1}\neq \emptyset$. 
        \item Suppose that $|P_i|>1$.  By inductive hypothesis, we can assume that $\mathcal{G}(P_i)$ is forced-connected. Note that $E(P_i\cap V_1)\subseteq P_i$. Hence once any play is in $P_i$, player 0 can visit  all the nodes in $P_i$ and then move from the set $P_i$ to $P_{i+1}$.
    \end{itemize}
    
    Therefore for all $Y,Z\in Son(X)$, player 0 can force the play from $Y$ to $Z$.
    Now constructing a winning strategy for player 0 is easy. Player 0 goes through the sets $C_1$, $C_2$, $\ldots$, $C_{m}$ in a circular way. Once a play enters $C_i$, player 0 forces the play to go through all the nodes in $C_i$ and then the player 
    moves from $C_i$ to $C_{i\text{ mod }m + 1}$.
\end{proof}

\begin{theorem}[{\bf Characterization Theorem}]\label{Thm:forced connected to SCC}
    The connectivity game $\mathcal{G}$ is forced-connected if and only if the directed graph $\mathcal{G}_{s}$ is strongly connected.
\end{theorem}

\begin{proof}
    $\Leftarrow$: Consider the forest $\Gamma(\mathcal{G})$ of $\mathcal{G}$. Since $\mathcal{G}_{s}$ is strongly connected, we have  $V\in  N$. By Lemma \ref{L:derivative node connectivity games}, $\mathcal{G}=\mathcal{G}(V)$ is forced-connected. 
    
    $\Rightarrow$: Assume that $\mathcal{G}$ is forced-connected and yet $|SCC(\mathcal{G}_{s})|>1$. Then by Lemma \ref{L:forced connected SCC G_k}, $\mathcal{G}_{s}\ne \mathcal{G}_{s+1}$. This is a contradiction.
\end{proof}

\section{Solving connectivity games efficiently}
\subsection{Efficiency results} 
In this section we state two theorems that solve connectivity games problem efficiently. 
The first theorem is the following. 

\begin{theorem}\label{Thm:|V_0|}
The connectivity game $\mathcal G$ can be solved in time   
$\mathbf{O}((\sqrt{|V_1|}+1) |E| + |V_1|^2)$. 
\end{theorem}

We point out two features of this theorem. The first is that if the cardinality $|V_1|$ is considered as a parameter, then we can solve the problem in linear time in $|V_0|$. The parameter constant in this case is $|V_1|^{3/2}$. The second feature is that the algorithm runs in linear time if the underlying game graph is sparse.  Our second theorem is the following:

\begin{theorem}\label{Thm:|V_1|}
The connectivity game $\mathcal G$ can be solved in time     $\mathbf{O}((|V_1|+|V_0|)\cdot |V_0|\log |V_0|)$. 
\end{theorem}

In comparison to the theorem above,  this theorem implies that we can solve the problem in linear time in $|V_1|$. The parameterised constant is $|V_0|\log{|V_0|}$. Furthermore,  the algorithm is more efficient than the first one on dense graphs. \ 

Finally, both of the algorithms outperform the standard known bound $\mathbf{O}(|V||E|)$ that solve the connectivity games.

It is easy to see that, by Theorem \ref{Thm:forced connected to SCC}, solving the connectivity games is related to  the incremental strongly connected component maintenance (ISCCM) problem.   Next we describe known solutions of the ISCCM problem, restate the problem suited to the game-theoretic setting, and prove the theorems. 

\subsection{Revisiting the ISCCM problem}\label{S:Addressing the ISCCM problem}

\smallskip

Finding SCCs of a digraph is a static version of strongly connected component maintenance problem. Call this  
the static SCC maintenance (SSCCM) problem. 
Tarjan's algorithm  solves the SSCCM problem in time $\mathbf{O}(m)$ \cite{tarjan1972depth}. In dynamic setting the  
 ISCCM problem is stated as follows. Initially, we are given $n$ vertices 
and the empty edge set. A sequence of edges $e_1, \ldots, e_m$ are added. No multiple edges and loops are allowed. 
The goal is to design a data structure that maintains 
the SCCs of the graphs after each addition of edges. 

We mention two algorithms that solve the ISCCM problem.  The first is  the {\em soft-threshold search algorithm}  
by Haeupler et al. \cite{haeupler2012incremental} that  handles sparse graphs. Their algorithm runs in time $\mathbf{O}(\sqrt{m}m)$. The second is by Bender et al. \cite{bender2009new,bender2011new}. Their algorithm is best suited for the class of dense graphs and runs in time of $\mathbf{O}(n^{2}\log n)$.  We employ the ideas of  these algorithms in the proofs of our 
Theorems \ref{Thm:|V_0|} and \ref{Thm:|V_1|} above.

We now revisit the ISCCM problem suited for our game-theoretic settings.  Initially we have $n$ vertices and  $m-k$ edges.
Edges $e_1, \ldots, e_k$ are added consecutively. No loops or multiple edges are allowed. After each addition of edges, the SCCs should be maintained.  We call this  the ISCCM$(m,k)$ problem. 
Observe that  the ISCCM$(m,k)$ problem can be viewed as the tradeoff between the SSCCM and the ISCCM problems. When $k=0$, the ISCCM$(m,0)$ problem coincides with the SSCCM problem; when $k=m$, we have the  original ISCCM problem.


\subsection{General frameworks}

\subsubsection{The first framework}\label{SS:first framework}
Theorem \ref{Thm:forced connected to SCC} reduces the connectivity game problem to the ISCCM$(m,k)$ problem.
Consider the construction of the game sequence $\mathcal{G}_{i}$. The construction focuses on the SCCs 
of $\mathcal{G}_{i}$.  The outgoing edges of player 1's nodes $v$ are added $\mathcal{G}_{i+1}$ iff $E(v)$ is a subset of the same SCC in $\mathcal{G}_i$. So, we simplify the sequence $\{\mathcal{G}_i\}_{i\ge 0}$ to $\{\mathcal{G}'_i\}_{i\ge 0}$: 

\begin{itemize}
    \item For $i=0$, let $F_0=\emptyset$, $U_0=V_1$ and $\mathcal{G}'_0=(V_0,V_1,E'_0)$, where $E'_0$ consists of outgoing edges in $E$ of player 0. 
    \item For $i>0$, \  $F_{i}=\{f_1,f_2,\ldots,f_m\}=$
     \begin{center}{$\bigcup_{S\in SCC(\mathcal{G}'_{i-1})} Force(U_{i-1},S), \  \mbox{and} \ 
    U_{i}=U_{i-1}\setminus F_{i}.
    $}
    \end{center} 
    Let $T'_i=\{(f_1,y_1),(f_2,y_2),\ldots,(f_m, y_m)\}$ be a set of edges consisting of exactly one outgoing edges of each $f\in F_{i}$. Then $\mathcal{G}'_i=(V_0,V_1,E'_i)$ where $E'_i=E'_{i-1}\cup T'_i$.
\end{itemize}

For this simplification we still have 
$SCC(\mathcal{G}_i)=SCC(\mathcal{G}'_i)$. Because of Theorem \ref{Thm:forced connected to SCC}, we have the following:

\begin{corollary}\label{C:forced connected to SCC}
    Player 0 wins the connectivity game on $\mathcal{G}$ if and only if $\mathcal{G}'_{s}$ is a SCC.\qed
\end{corollary}

Now we describe our algorithm for the connectivity game problem. Initially, check  if 
$\mathcal{G}$ is strongly connected.  If not, return false. Otherwise, start with $\mathcal{G}'_0$, and apply an
ISCCM$(m,k)$ algorithm to the sequence $\mathcal G'_0, \mathcal G'_1, \ldots$ to maintain the SCCs. At stage $i$, a vertex is forced if all its outgoing edges go into the same SCC in $\mathcal{G}'_i$.  The algorithm collects 
all forced vertices $u\in U_i$. 
Then $T'_i$ is the set of edges consisting of exactly one outgoing edge for  each forced vertex $u$.  Then all  edges in $T'_i$ are added 
to $\mathcal{G}'_i$ by running the ISCCM$(m,k)$ algorithm. 
 The construction runs until all vertices $u\in U_i$ aren't forced. Finally, check if $\mathcal{G}'_s$ is strongly connected.
 

To implement the algorithm efficiently, we construct the sequence $\mathcal G'_0, \mathcal G'_1, \ldots$ iteratively. At stage $i$, we maintain $\mathcal{G}'=\mathcal{G}'_i$, $U=U_{i}$ and compute $T=T_i$. We maintain $U$ and outgoing edges of $u\in U$ in $\mathcal{G}$ using singly linked lists. Let \textit{first-U} be the first vertex in $U$ and \textit{next-U}$(u)$ be the successor of $u$ in $U$. For each $u\in U$, we maintain a list of its outgoing edges. Let \textit{first-out-U}$(u)$ be the first edge on $u$'s outgoing list and \textit{next-out-U}$((u,x))$ be the edge after $(u,x)$. We maintain $T$ using singly linked lists. Let 
\textit{first-T} be the first edge in $T$ and \textit{next-T}$((f,y))$ be the successor of $(f,y)$ in $T$. The list $T$ 
is initialized to $\emptyset$ at the start of each stage. When collecting forced vertices from $U$, the vertices  are examined sequentially. Assume $u\in U$ is being examined. Let $(u,x)$=\textit{first-out-U}$(u)$. If \textit{next-out-U}$((u,x))$=\textit{null} then $u$ is forced and edge $(u,x)$ is added into $T$. Otherwise, let $(u,y)$=\textit{next-out-U}$((u,x))$. Let 
\textit{the-same-scc}$(x,y)$ be the function that checks if $x$ and $y$ are in the same SCC of $\mathcal{G}'$. 
 If so, 
 then \textit{the-same-scc}$(x,y)$=\textit{true}. Otherwise, \textit{the-same-scc}$(x,y)$=\textit{false}. Call this the  {\em SCC test}. If \textit{the-same-scc}$(x,y)$=\textit{true}, then set \textit{first-out-U}$(u)$=$(u,y)$  and move 
 to  the next outgoing edge from $u$. Otherwise, process the next vertex in $U$. After examining all $u\in U$, the new $U$ is $U_{i+1}$ and $T$ is $T_i$. If $T=\emptyset$  then  $\mathcal{G}'=\mathcal G_s'$.  Otherwise, all edges in $T$ are added to $\mathcal{G}'$ and the new $\mathcal{G}'$ is $\mathcal{G}'_{i+1}$. 
 We present this implementation  
 in Figure \ref{F:DFCG-ISCCM-M-K}, and the process  is called the DFCG-M-K$(V_0,V_1,E)$  function.


\begin{figure}[H]
    \centering 
    \scriptsize
    \begin{tabular}{l}
    bool \textbf{function} DFCG-M-K(vertex set $V_0$, vertex set $V_1$, edge set $E$)\\
        \hspace*{4mm}Run a linear-time algorithm to check the strong connectivity of $\mathcal{G}=(V_0,V_1,E)$\\
        \hspace*{4mm}\textbf{if} $\mathcal{G}$ isn't strongly connected \textbf{then return} \textit{false}\\
        \hspace*{4mm}$E'=\{(u,v)\mid u\in V_0 \text{ and } (u,v)\in E\}$;\\
        \hspace*{4mm}Initialize \textit{first-U} and \textit{next-U} with $V_1$, \\
        \hspace*{3mm}and \textit{first-out-U}  and \textit{next-out-U} with $V_1$'s outgoing lists\\
        \hspace*{4mm}Initialize an ISCCM$(m,k)$ algorithm with $\mathcal{G}'=(V_0,V_1,E')$\\
        \hspace*{4mm}\textbf{repeat}\\
            \hspace*{8mm}$u=$\textit{first-U}; \textit{previous-u}=\textit{null}; \textit{first-T}=\textit{null}; \textit{last-T}=\textit{null}\\
            \hspace*{8mm}\textbf{while} $u\ne$\textit{null} \textbf{do}\\
                    \hspace*{12mm}$(u,x)$=\textit{first-out-U}$(u)$\\
                    \hspace*{12mm}\textbf{if} \textit{next-out-U}$((u,x))$=\textit{null} \textbf{then}\\
                        \hspace*{16mm}\textbf{if} \textit{first-T}=\textit{null} \textbf{then} \textit{first-T}=$(u,x)$; \textit{last-T}=$(u,x)$\\
                        \hspace*{16mm}\textbf{else} \textit{next-T}(\textit{last-T})=$(u,x)$; \textit{last-T}=$(u,x)$\\
                        \hspace*{16mm}\textit{next-T}$((u,x))$=\textit{null}\\
                        \hspace*{16mm}\textbf{if} \textit{first-U}=$u$ \textbf{then} \textit{first-U}=\textit{next-U}$(u)$\\
                        \hspace*{16mm}\textbf{if} \textit{previous-u}$\ne$\textit{null} \textbf{then} \textit{next-U}(\textit{previous-u})=\textit{next-U}$(u)$\\
                        \hspace*{16mm}$u$=\textit{next-U}$(u)$\\
                    \hspace*{12mm}\textbf{else}\\
                        \hspace*{16mm}$(u,y)$=\textit{next-out-U}$((u,x))$ \\
                        \hspace*{16mm}\textbf{if} \textit{the-same-scc}$(x,y)$=\textit{true} \textbf{then} \textit{first-out-U}$(u)$=$(u,y)$\\
                        \hspace*{16mm}\textbf{else} \textit{previous-u}=$u$; $u$=\textit{next-U}$(u)$\\
                    \hspace*{12mm}\textbf{end}\\
            \hspace*{8mm}\textbf{end}\\
            \hspace*{8mm}$T$=\textit{first-T}\\
            \hspace*{8mm}\textbf{while} $T\ne$\textit{null} \textbf{do}\\
                \hspace*{12mm}$(f,y)=T$; $T$=\textit{next-T}$(T)$\\
                \hspace*{12mm}Add edge $(f,y)$ to $\mathcal{G}'$ by running ISCCM$(m,k)$ algorithm\\
            \hspace*{8mm}\textbf{end}\\
        \hspace*{4mm}\textbf{until} \textit{first-T}=\textit{null}\\
        \hspace*{4mm}Run a linear-time algorithm to check the strong connectivity of $\mathcal{G}'$\\
        \hspace*{4mm}\textbf{if} $\mathcal{G}'$ is strongly connected \textbf{then return} \textit{true}\\
        \hspace*{4mm}\textbf{else return} \textit{false}

    \end{tabular}
    \caption{Implementation of DFCG-M-K function}
    \label{F:DFCG-ISCCM-M-K}
\end{figure}

\begin{theorem}\label{Thm:total charge of DFCG-M-K algorithm}
    The total running time of the DFCG-M-K algorithm is bounded by  the running times for  (1) the ISCCM$(|E|,|V_1|)$ algorithm, 
    (2) $|E|+|V_1|^{2}$  many SCC tests,  and (3) an extra $\mathbf{O}(|E|+|V_1|)$ running time. 
\end{theorem}

\begin{proof}
    Initially, the DFCG-M-K algorithm checks strong connectivity of $\mathcal{G}$ and $\mathcal{G}'$. This takes $\mathbf{O}(|E|)$ time. \ Observe that $\mathcal{G}'$ is initialized to the acyclic digraph with at most $|V_1|$ edge additions. Therefore, the maintenance of SCCs is an ISCCM$(m,k)$ problem where $m\leq |E|$ and $k\leq |V_1|$.

    Since $s\leq |V_1|$, where $s$ is the stabilization point, there are at most $|V_1|$ stages. At each stage, vertices in $U$ are examined. Consider the examination of a vertex $u\in U$.  If $u$ has one outgoing edge, then $u$ is forced and removed from $U$. This takes $\mathbf{O}(|V_1|)$ time. If $u$ has multiple outgoing edges then there is at least one SCC test. If there are multiple SCC tests then the extra SCC tests correspond to the deletions of outgoing edges. Since there are at most $|E|$ edges to delete, the total number of SCC tests is bounded by $|E|+|V_1|^{2}$.
\end{proof}

\subsubsection{The second framework}\label{SS:second framework}

We analyse  $\{\mathcal{G}'_k\}_{k\ge 0}$ further. As each vertex in $V_1$ has at most one outgoing edge in $\mathcal{G}'_k$,  
we can remove player 1's vertices $v$ by adjoining ingoing edges into $v$ with the outgoing edge from $v$.
This reduces the size of $\mathcal{G}'_k$ and preserves the strong connectedness.  Here is the process: 

\begin{itemize}
    \item For $k=0$, let $F''_0=\emptyset$, $U''_0=V_1$ and $\mathcal{G}''_0=(V_0,E''_0)$, where $E''_0=\emptyset$. 
    \item For $k>0$, set 
    $F''_{k}=\bigcup_{S\in SCC(\mathcal{G}''_{k-1})} Force(U''_{i-1},S)$, $U''_{k}=U''_{k-1}\setminus F''_{k}$, and let
    $T''_{k}$ be the set of edges $\{(f_1,y_1),(f_2,y_2),\ldots,(f_m, y_m)\}$ consisting of exactly one outgoing edge for each $f\in F''_{k}$. Then $\mathcal{G}''_k=(V_0,E''_k)$ where $E''_k=E''_{k-1}\cup \bigcup_{(f,y)\in T''_{k}}(E^{-1}(f)\times \{y\})$.
\end{itemize}

\noindent

\begin{lemma} \label{Lem:weak connectivity G' and G''}
    For all $k\ge 0$ and $v,u\in V_0$, there exists a path from $v$ to $u$ in $\mathcal{G}'_{k}$ iff there exists a path from $v$ to $u$ in $\mathcal{G}''_{k}$. 
\end{lemma}

\begin{proof}
    If $k=0$ then $F_0=F''_0$, $U_0=U''_0$ and there are no paths between player 0's nodes in both $\mathcal{G}'_{0}$ and $\mathcal{G}''_{0}$, which holds the lemma. For $k>0$ our hypothesis is the following.
    \begin{itemize}
        \item For all $i<k$, $F_i=F''_i$ and $U_i=U''_i$.
        \item For all $i<k$ and $v,u\in V_0$, $(v,u)\in E''_i$ iff there is a $w\in \bigcup_{0\le j\le i}F_j$ such that $(v,w)$ and $(w,u)$ are in $E'_i$.
        \item For all $i<k$, and $v,u\in V_0$, there exists a path from $v$ to $u$ in $\mathcal{G}'_{i}$ iff there exists a path from $v$ to $u$ in $\mathcal{G}''_{i}$.
    \end{itemize}
    Therefore, 
    $$
    SCC(\mathcal{G}''_{k-1})=\{X\cap V_0\mid X\in SCC(\mathcal{G}'_{k-1}) \  \mbox{and}  \ X\cap V_0\ne \emptyset\}.
    $$ 
    Since $E\subseteq(V_0\times V_1)\cup (V_1\times V_0)$, for all $U\subseteq V_1$ and $S\subset V$ with $U\cap S =\emptyset$, $Force(U,S)=Force(U,S\cap V_0)$. Note that if $S\in SCC(\mathcal{G}'_{k-1})$ and $S\cap V_0 \ne \emptyset$ then $U_{k-1}\cap S=\emptyset$; if $S\in SCC(\mathcal{G}'_{k-1})$ and $S\cap V_0 = \emptyset$ then $Force(U_{k-1}, S)=\emptyset$. Therefore, $F_k=F''_{k}$,
   $U_{k}=U''_k$. Since $F_{k}=F''_{k}$, we can set $T'_k=T''_k$. Hence for all $v,u\in V_0$, $(v,u)\in \bigcup_{(f,y)\in T''_{k}}(E^{-1}(f)\times \{y\})$ iff there is a $w\in F_k$ such that $(v,w)$ and $(w,u)$ are in $E'_k$. Since $E''_k=E''_{k-1}\cup \bigcup_{(f,y)\in T''_{k}}(E^{-1}(f)\times \{y\})$,  for all $v,u\in V_0$ we have  $(v,u)\in E''_{k}$ iff there is a $w\in \bigcup_{0\le j\le k}F_j$ such that $(v,w), (w,u)\in E'_{k}$. Hence, for all $v,u\in V_0$ there is a path from $v$ to $u$ in $\mathcal{G}'_{k}$ iff there is a path from $v$ to $u$ in $\mathcal{G}''_{k}$. 
\end{proof}


\begin{lemma}\label{Lem:G' SCC to G'' and G SCC}
    $\mathcal{G}'_{s}$ is a SCC iff both $\mathcal{G}''_{s}$ and $\mathcal{G}$ are SCC. 
\end{lemma}

\begin{proof}
    $\Rightarrow$: Assume $\mathcal{G}'_{s}$ is a SCC. Since $E'_{s}\subseteq E$, $\mathcal{G}$ is strongly connected. By Lemma \ref{Lem:weak connectivity G' and G''}, $\mathcal{G}''_{s}$ is strongly connected.

    $\Leftarrow$: Assume both $\mathcal{G}''_{s}$ and $\mathcal{G}$ are strongly connected. By Lemma \ref{Lem:weak connectivity G' and G''}, player 0's vertices are strongly connected in $\mathcal{G}'_{s}$. Since $\mathcal{G}$ is strongly connected, all player 1's vertices have incoming edges from player 0's vertex  and outgoing edges to player 0's vertex. Therefore, $\mathcal{G}'_{s}$ is strongly connected.
\end{proof}

\noindent
By Corollary \ref{C:forced connected to SCC} and Lemma \ref{Lem:G' SCC to G'' and G SCC}, we have the following:

\begin{corollary}\label{C:forced connected to SCC in G and G''}
    Player 0 wins the connectivity game on $\mathcal{G}$ if and only if both $\mathcal{G}''_{s}$ and $\mathcal{G}$ are strongly connected.
\end{corollary}

Now we describe our algorithm that solves the connectivity game problem. Initially, check (in linear time) if 
$\mathcal{G}$ is strongly connected.  If not, return false. Else, start with $\mathcal{G}''_0$, and apply an
ISCCM algorithm to the sequence $\mathcal G''_0, \mathcal G''_1, \ldots$ to maintain the SCCs. During the process, 
compute $T''_1,T''_2,\ldots$. Initially, $T''_1=\{(v,u)\mid v\in V_1 \text{ and } E(v)=\{u\}\}$. At stage $i$, all edges in $\bigcup_{(f,y)\in T''_{i}}(E^{-1}(f)\times \{y\})$ are added to $\mathcal{G}''_{i-1}$ by running the ISCCM algorithm. The set 
$T''_{i+1}$ is computed when two SCCs in ISCCM algorithm are joined.
 The construction runs until $T''_{s+1}=\emptyset$. Finally, check if $\mathcal{G}''_s$ is strongly connected.

To implement the algorithm, construct $\mathcal G''_0, \mathcal G''_1, \ldots$ iteratively. At stage $i$, maintain $\mathcal{G}''=\mathcal{G}''_{i-1}$, $T=T''_{i}$, and compute $\mathcal{G}''_i$ and $T''_{i+1}$. We maintain $T$ using singly linked lists. Let \textit{first-T} be the first edge $(f,y)$ in $T$ and \textit{next-T}$((f,y))$ be the successor of $(f,y)$.   For $v\in V_1$, let \textit{outdegree}$(v)$ be the outdegree of $v$ in $\mathcal{G}$, and
for the canonical vertex $u\in \mathcal{G}''$, we let \textit{indegree}$(u,v)$ be the number of outgoing edges from $v$ in $\mathcal G''$ to the SCC identified by canonical vertex $u$. 
The tables \textit{indegree} and \textit{outdegree} are initialized by outgoing edges of player 1's vertices. At each stage, we traverse $T$ by variable $t$ and then set \textit{first-T}=\textit{last-T}=\textit{null} to build the new $T$ for the next stage. In each traversal $t=(f,y)$, add all edges in $E^{-1}(f)\times \{y\}$ to $\mathcal{G}''$ by running the ISCCM algorithm. When two SCCs are combined, \textit{indegree} is updated and if some player 1's vertex becomes forced, then it is put into $T$. Let UPDATE-INDEGREE$(x,y)$ be the function that updates \textit{indegree} and puts new forced vertices into $T$ where $x,y$ are  canonical vertices and $x$ is the new canonical vertex. UPDATE-INDEGREE is an auxiliary macro  in Figure \ref{F:UPDATE-INDEGREE}, intended to be expanded in-line. The macro is called when two SCCs are joined in the ISCCM algorithm. Once a forced vertex $v$ is in $T$, \textit{outdegree}$(v)$ is set to -1, and  $v$ will never be added to $T$ again. 
After adding all edges, $\mathcal{G}''=\mathcal{G}''_{i}$ and $T=T''_{i+1}$.  If $T=\emptyset$ then $\mathcal{G}''=\mathcal{G}''_s$. Otherwise,  the next stage is considered. We call the process 
the DFCG$(V_0,V_1,E)$ function presented in Figure \ref{F:DFCG-ISCCM-M-M}.
 

\begin{figure}[H]
    \centering 
    \scriptsize
    \begin{tabular}{l}
        \textbf{macro} UPDATE-INDEGREE(vertex $x$, vertex $y$)\\ 
        \hspace{4mm}\textbf{for} $v\in V_1$ \textbf{do}\\
            \hspace{8mm}\textit{indegree}$(x,v)$=\textit{indegree}$(x,v)$+\textit{indegree}$(y,v)$\\
            \hspace{8mm}\textbf{if} \textit{outdegree}$(v)$=\textit{indegree}$(x,v)$ \textbf{then}\\
                \hspace{12mm}\textit{outdegree}$(v)$=-1\\
                \hspace{12mm}\textbf{if} \textit{first-T}=\textit{null} \textbf{then} \textit{first-T}=$(v,x)$; \textit{last-T}=$(v,x)$\\
                \hspace{12mm}\textbf{else} \textit{next-T}(\textit{last-T})=$(v,x)$; \textit{last-T}=$(v,x)$\\
            \hspace{8mm}\textbf{end}\\
        \hspace{4mm}\textbf{end}\\
    \end{tabular}
    \caption{Implementation of UPDATE-INDEGREE function}
    \label{F:UPDATE-INDEGREE}
\end{figure}
\vspace{-2mm}
\begin{figure}[H]
    \centering 
    \scriptsize
    \begin{tabular}{l}
        bool \textbf{function} DFCG(vertex set $V_0$, vertex set $V_1$, edge set $E$)\\
        \hspace*{4mm}Run a linear-time algorithm to check the strong connectivity of $\mathcal{G}=(V_0,V_1,E)$\\
        \hspace*{4mm}\textbf{if} $\mathcal{G}$ isn't strongly connected \textbf{then return} \textit{false}\\
        \hspace{4mm}Initialize an ISCCM algorithm with $\mathcal{G}''=(V_0,E''=\emptyset)$\\
        \hspace{4mm}Initialize \textit{outdegree}$(v)$=0 and \textit{indegree}$(u,v)$=0 for all $v\in V_1$ and $u\in V_0$\\
        \hspace{4mm}\textbf{for} $v\in V_1$ \textbf{do}\\
            \hspace{8mm}\textbf{for} $u\in E(v)$ \textbf{do}\\
                \hspace{12mm}\textit{outdegree}$(v)$=\textit{outdegree}$(v)$+1; \textit{indegree}$(u,v)$=1\\
        \hspace{4mm}\textit{first-T}=\textit{null}; \textit{last-T}=\textit{null}\\
        \hspace{4mm}\textbf{for} $v\in V_1$ \textbf{do}\\
        \hspace{8mm}\textbf{if} \textit{outdegree}$(v)$=1 \textbf{then}\\
            \hspace{12mm}\textit{outdegree}$(v)$=-1; $\{u\}$=$E(v)$\\
            \hspace{12mm}\textbf{if} \textit{first-T}=\textit{null} \textbf{then} \textit{first-T}=$(v,u)$; \textit{last-T}=$(v,u)$\\
            \hspace{12mm}\textbf{else} \textit{next-T}(\textit{last-T})=$(v,u)$; \textit{last-T}=$(v,u)$\\
        \hspace{8mm}\textbf{end}\\
        \hspace{4mm}\textbf{end}\\
        \hspace*{4mm}\textbf{while} \textit{first-T}$\ne$ \textit{null} \textbf{do}\\
            \hspace*{8mm}$t=$\textit{first-T}; \textit{first-T}=\textit{null}; \textit{last-T}=\textit{null}\\
            \hspace*{8mm}\textbf{while} $t\ne$ \textit{null} \textbf{do}\\
                \hspace{12mm}$(f,y)$=$t$; $E''=E^{-1}(f)\times \{y\}$; $t=$\textit{next-T}$(t)$ \\
                \hspace{12mm}Add edges in $E''$ to $\mathcal{G}''$ by running ISCCM algorithm \\
                \hspace{11mm}with UPDATE-INDEGREE\\
            \hspace*{8mm}\textbf{end}\\
        \hspace*{4mm}\textbf{end}\\
        \hspace*{4mm}Run a linear-time algorithm to check the strong connectivity of $\mathcal{G}''$\\
        \hspace*{4mm}\textbf{if} $\mathcal{G}''$ is strongly connected \textbf{then return} \textit{true}\\
        \hspace*{4mm}\textbf{else return} \textit{false}

    \end{tabular}
    \caption{Implementation of DFCG function}
    \label{F:DFCG-ISCCM-M-M}
\end{figure}

\begin{theorem}\label{Thm:total charge of DFCG-ISCCM-M-M algorithm}
    The total time of the DFCG algorithm is bounded by  (1) the time for an ISCCM algorithm with $|V_0|$ vertices and $|V_0|^2$ edge additions, together with  (2) $\mathbf{O}(|V_0|\cdot|V_1|\log |V_0|)$.
     \end{theorem}

\begin{proof}
    Note that $|E|\leq 2 \cdot |V_0|\cdot |V_1|$ and $|E''_s|\le |E|$. The DFCG algorithm starts with checking the strong connectivity of $\mathcal{G}$ and ends with checking the strong connectivity of $\mathcal{G}''$. These take $\mathbf{O}(|V_0|\cdot |V_1|)$ time.

    The graph $\mathcal{G}''$ is initialized with $|V_0|$ vertices and no edges. The construction adds  
    at most $|V_0|\cdot |V_1|$ edges  to $\mathcal{G}''$. There are at most $|V_0|^2$ edge additions which may change the connectivity of $\mathcal{G}''$, and each of the remaining duplicated edge addition takes $\mathbf{O}(1)$ time. So, the total time of maintaining SCCs consists of (1) $\mathbf{O}(|V_0|\cdot |V_1|)$ time and (2) time of an ISCCM algorithm with $|V_0|$ vertices and $|V_0|^2$ edge additions.

    Consider computations of \textit{outdegree} and \textit{indegree}. Since entries are within $[-1,|V_0|]$, 
    each computation on the entry takes $\mathbf{O}(\log |V_0|)$ time. The  \textit{outdegree} and \textit{indegree} are initialized by player 0's outgoing edges. This takes $\mathbf{O}(|V_0|\cdot |V_1|\log |V_0| )$ time. Since SCC are joined at 
    most $|V_0|-1$ times, UPDATE-INDEGREE is called at most $|V_0|-1$ times. During the initialization of $T$ or calling UPDATE-INDEGREE, there are $|V_1|$ additions and $|V_1|$ compare operations. This takes $\mathbf{O}(|V_0|\cdot|V_1|\log |V_0|)$ time in total. 

    Consider the maintenance of $T$.  When initializing $T$ or calling UPDATE-INDEGREE, all player 1's vertices are traversed to find a forced vertex. This takes $\mathbf{O}(|V_0|\cdot|V_1|)$ time. Each player 1's vertex is added into $T$ once it becomes forced and then it will never enter $T$ again. After a vertex enters $T$, it will be traversed one time and then removed from $T$. This takes $\mathbf{O}(|V_1|)$ time. 
\end{proof}

\subsection{Proofs of Theorem \ref{Thm:|V_0|} and Theorem \ref{Thm:|V_1|}}

The proof of Theorem \ref{Thm:|V_0|} uses the first framework together with the ISCCM$(m,k)$ algorithm presented in Section \ref{S:ISCCM(m,k)}. The proof of Theorem  \ref{Thm:|V_1|} uses the second framework. 

\begin{restateTheorem}{Thm:|V_0|}
    The connectivity game $\mathcal G$ can be solved in time  $\mathbf{O}((\sqrt{|V_1|}+1) |E| + |V_1|^2)$.
\end{restateTheorem}
\begin{proof}
    Apply the DFCG-M-K process in Section \ref{SS:first framework} to solve $\mathcal G$ with the ISCCM$(m,k)$ algorithm presented in  Section \ref{S:ISCCM(m,k)}. The decision condition \textit{the-same-scc}$(x,y)$=\textit{true} is implemented by \textit{find}$(x)$=\textit{find}$(y)$. The \textit{find}$(x)$  outputs the canonical vertex of the SCC of $x$. The \textit{find} is maintained in soft-threshold search algorithm and the time per find is $\mathbf{O}(1)$. The soft-threshold search algorithm is run iff $\mathcal{G}$ is strongly connected. So, the soft-threshold search is run with $|E|\ge|V|$. By Theorems \ref{Thm:total charge of DFCG-M-K algorithm}  and \ref{Thm:time complexity of SCC soft-threshold search}, and Corollary \ref{C:forced connected to SCC},  the theorem is proved. 
\end{proof}

\begin{restateTheorem}{Thm:|V_1|}
    There exists an algorithm that solves any given connectivity game $\mathcal G$ in time  $\mathbf{O}((|V_1|+|V_0|)\cdot |V_0|\log |V_0|)$. 
\end{restateTheorem}
\begin{proof}
    Apply the DFCG algorithm in Section \ref{SS:second framework} to solve the connectivity game where the ISCCM algorithm is implemented by the $\mathbf{O}(n^2\log n)$ time solution of Bender et al. \cite{bender2009new,bender2011new}. By Theorem \ref{Thm:total charge of DFCG-ISCCM-M-M algorithm} and Corollary \ref{C:forced connected to SCC in G and G''}, the theorem is proved.
\end{proof}

\section{The ISCCM$(m,k)$ algorithm} \label{S:ISCCM(m,k)} 



In this section we adapt  (into our setting) the algorithm provided by Haeupler et al. in \cite{haeupler2012incremental}  that solves 
the ISCCM problem and analyse it. 
Our strategy of maintaining the SCCs is this. Apply a SSCCM algorithm on the initial graph with $m-k$ edges to compute SCCs, list canonical vertices (representing the SCCs) in a topological order, and keep the relationships between canonical vertices. This produces an acyclic graph with at most $n$ 
vertices and $m-k$ edges. Use this graph to initialise any of the algorithms that solves the ISCCM problem and then run the selected algorithm on the $k$ addition of edges. The key is that the initial $m-k$ edges might always 
be used in running the algorithm. Hence, the algorithm might be dependent on all the $m$ edges. We want to reduce this dependency. 
We refine the algorithm of Haeupler et al. \cite{haeupler2012incremental}  so that the algorithm runs in linear time if $k$ is taken as a parameter, and the parameter constant is $\sqrt{k}$.

\smallskip

In solving  the  ISCCM$(m,k)$ problem three issues arise: (1) cycle detection, (2) topological order maintenance  and (3) maintenance of the SCCs. The compatible search algorithm from \cite{haeupler2012incremental} 
 (see Section \ref{S:soft-threshold search algorithm for ISCCM problem}, Appendix)  detects cycles and maintains the topological order after each edge $(v,w)$ addition. Recall that a topological order on 
a digraph is a total order "$<$" of the vertices such that for each edge $(v,w) \in E$ we have $v<w$. The algorithm uses the notion of {\em related edges}. Two distinct edges in a graph {\em are related} if they are on a common path. 
The compatible search algorithm applied in our setting of the ISCCM$(m,k)$ problem satisfies the following lemma:


\begin{lemma}\label{Lem: edge traversals of k edge additions}
    Given an acyclic graph with $m-k$ edges ($m\ge n$ and $m\ge k\ge 0$), the  compatible search algorithm 
    does $\mathbf{O}(\sqrt{k}m)$ edge traversals over $k$ edge additions.
\end{lemma}

\begin{proof}
We assume that $k>0$ as $k=0$ case is trivial. Consider the edge addition process $e_1, e_2, \ldots$. The algorithm performs forward and backward searches 
when added edge does not agree with the topological order. \ Let $e_{t+1}=(v_{t+1},w_{t+1})$ be the first edge addition that creates a cycle. The compatible search algorithm does  $\mathbf{O}(m)$ edge traversals at this stage. 

\smallskip

Let $s=\{s_1,s_2,\ldots,s_{t}\}$ be the set of numbers such that $s_i$ is the number of edges traversed during the forward search of $i$th compatible search, where $1\leq i\leq t$. In  \cite{haeupler2012incremental}  it is proved that the edge addition of $e_i$, $1\leq i\leq t$, increases the number of related edge pairs by at least $s_i(s_i+1)/2$. Also, any edge traversed forward has a distinct twin traversed backward during the same search step.


\smallskip

    Call a search at stage $i$, $i\leq t$,  {\em small} if it does no more than $2m/\sqrt{k}$ edge traversals and {\em big} otherwise. Since there are at most $k$ small searches, together they do at most $2\sqrt{k}m$ edge traversals. Let 
    $$
    D=\{d_1,d_2,\ldots,d_q\}=\{i\mid i\in [1,t]\text{ and }s_{i}>m/\sqrt{k}\}.
    $$ 
    Thus, $D$ contains indices  $i$ of all $s_i$ such that the search at stage $i$ is big. All big searches add at least $\sum_{i=1}^{q} s_{d_i}(s_{d_i}+1)/2>m/\sqrt{k}\sum_{i=1}^{q} s_{d_i}/2$ many related edge pairs. Since there are at most $\binom{m}{2}<m^{2}/2$ related edge pairs, we have $m^{2}/2>m/\sqrt{k}\sum_{i=1}^{q} s_{d_i}/2$. Hence, $\sum_{i=1}^{q} s_{d_i}<\sqrt{k}m$. Therefore, all the big searches do at most $2\sqrt{k}m$ edge traversals.
\end{proof}

To implement the compatible search algorithm efficiently, we utilise  the soft-threshold search algorithm  
from \cite{haeupler2012incremental} (Section  \ref{S:soft-threshold search algorithm for ISCCM problem}, Appendix). 

\begin{lemma} \label{Lem:time complexity of soft-threshold search}
    Given an acyclic graph with $m-k$ edges ($m\ge n$ and $m\ge k\ge 0$), and its topological order, the soft-threshold search takes $\mathbf{O}((\sqrt{k} + 1)m)$ time over the initialization and $k$ edge additions.
\end{lemma}
\begin{proof}
    Given acyclic graph and topological order, the initialization of soft-threshold search takes $\mathbf{O}(m)$ time.
  The total charge of soft-threshold search is $\mathbf{O}(1)$ time per edge traversal (see \cite{haeupler2012incremental}).   
   With Lemma \ref{Lem: edge traversals of k edge additions}, we are done. \end{proof}

We use disjoint set data structure \cite{tarjan1975efficiency,tarjan1984worst} to represent the vertex partition defined by SCCs. This supports the query \textit{find}$(v)$ which returns the canonical vertex,  and the operation \textit{unite}$(x,y)$, which given canonical vertices $x$ and $y$, joins the components of $x$ and $y$, and makes $x$  a new 
canonical vertex. 

The soft-threshold search maintains the canonical vertices.  When an edge addition combines several SCCs, we combine the incoming and outgoing lists, and remove the non-canonical vertices. This may create multiple edges between the same pair of SCCs and loops. 
We delete the loops during the search.

\begin{lemma}\label{Lem: SCC edge traversals of k edge additions}
    Given an acyclic graph with $m-k$ edges ($m\ge n$ and $m\ge k\ge 0$) and its topological order, with a sequence of 
    $k$ edge additions, maintaining SCCs via soft-threshold search does $\mathbf{O}(\sqrt{k}m)$ edge traversals.
\end{lemma}

\begin{proof}
    We assume that $k>0$ as $k=0$ case is trivial. Divide the edge traversals during a search into those of edges that become loops, and those that do not. Note that an edge becomes a loop only if some edge addition triggers search.
    Over all $\mathbf{O}(m)$ edges turn into loops. 
       
    Suppose the addition of $(v,w)$ triggers a search. Let $(u,x)$ and $(y,z)$ be edges traversed during 
    forward and backward search, respectively,  such that $\textit{find}(u) <\textit{find}(z)$. In  \cite{haeupler2012incremental}  it is proved that either $(u,x)$ and $(y,z)$ are unrelated before the addition of $(v,w)$ but related afterwards, or they are related before the addition and the addition makes them into loops. The proof of Lemma \ref{Lem: edge traversals of k edge additions} implies that there are $\mathbf{O}(\sqrt{k}m)$ traversals of edges that do not turn into loops.
\end{proof}

We apply the lemmas above to solve the ISCCM$(m,k)$ problem.
The algorithm is in Section  \ref{S:soft-threshold search algorithm for ISCCM problem}, Appendix.


\begin{theorem}\label{Thm:time complexity of SCC soft-threshold search}
  The    ISCCM$(m,k)$ problem can be solved in  
  $\mathbf{O}((\sqrt{k}+1)m+\min(k+1,n)\log (\min(k+1,n)))$ time.
\end{theorem}

\begin{proof}
Thus, our solution of the ISCCM$(m,k)$ problem consists of 3 steps. In the first step,  we  use Tarjan's algorithm to find the SCCs of the initial graph with $m-k$ edges. In the second step, we initialise the ISCCM algorithm with the output of the first step.
In the third step,  we run the ISCCM algorithm on remaining $m-k$ edge additions. The first two steps take  $\mathbf{O}(m)$ time in total. Then we apply the soft-threshold search algorithm as the ISCCM algorithm. Now we analyse this process.

Assume that during the process an edge  $(v,w)$ is added such that $find(v)>find(w)$. Each search step either traverses two edges or deletes one or two loops. An edge can only become a loop once and be deleted once.  So the time for such events is $\mathbf{O}(m)$ over all edge additions. The edges in $Y$ (see \ref{S:soft-threshold search algorithm for ISCCM problem}) are 
traversed by the search so that the time to form the new component and reorder the vertices is $\mathbf{O}(1)$ per edge traversal. By Lemma \ref{Lem: SCC edge traversals of k edge additions} and the proof of Lemma \ref{Lem:time complexity of soft-threshold search}, the ISCCM algorithm part takes $\mathbf{O}((\sqrt{k}+1)m)$ time plus $\mathbf{O}(\min(k+1,n)\log (\min(k+1,n)))$ time (derived from the disjoint set data structure) in total. 
\end{proof}

\section{Explicitly given  M\"uller games}
\subsection{Full proof of Horn's algorithm}\label{SS:proof of Horn}
We start with standard notions about  games on graphs. 
Let $\mathcal{G}$ be a M\"uller game.  A set $S\subseteq V$ determines a subgame in $\mathcal{G}$ if for all $v\in S$ we have  $E(v)\cap S\ne \emptyset$. We call $\mathcal{G}(S)$, the subgame of $\mathcal{G}$ determined by $S$. The set $S\subseteq V$ is a $\sigma$-trap in $G$ if $S$ determines a subgame in $\mathcal{G}$ and $E(S\cap V_\sigma)\subseteq S$.


By  $Win_\sigma(\mathcal{G})$ we denote the set of all vertices $v$ in $\mathcal{G}$ such that player $\sigma$ wins game $\mathcal{G}$ starting from $v$.

Let $Attr_{\sigma}(X, G(Y))$ be a set of all vertices $v$ in $Y$ such that player $\sigma$ can force the token from $v$ to $X$ in game $\mathcal G(Y)$. 

Let $\Omega$ be the set of all winning sets of the M\"uller game $\mathcal{G}$. We can {\em 
topologically linear order $<$} the set $\Omega$, that is, 
for all distinct $X,Y\in \Omega$, if $X\subsetneq Y$ then $X<Y$. Thus, if $W_1< W_2< \ldots < W_s$ is a topological linear order on $\Omega$ then we have the following. If $i<j$ then $W_i\not\supseteq W_j$.

Below we provide several results that are interesting on their own. We will also use them in our analysis of M\"uller games. 

\begin{lemma}\label{L:attract or wins}
    Let $\mathcal{G}$ be a game, $F_0=\Omega$ and $F_1=2^{V}\setminus \Omega$. If $V\in F_\sigma$ and for all $v\in V$, either $Attr_\sigma(\{v\}, G)=V$ or player $\sigma$ wins $\mathcal{G}(V\setminus Attr_\sigma(\{v\}, G))$ then player $\sigma$ wins $\mathcal{G}$.
\end{lemma}

\begin{proof}
    Let $V=\{v_1,v_2,\ldots,v_n\}$. We construct a winning strategy for player $\sigma$ as follows. Starting at $i=1$, 
    Player $\sigma$ considers $v_{i}$. 
    \begin{itemize}
        \item If the token is in $Attr_\sigma(\{v_i\}, G)$ then player $\sigma$ forces the token to $v_i$ and then considers next vertex $v_{i \text{ mod }n + 1}$.
        \item Otherwise player $\sigma$ follows the winning strategy in the game $\mathcal{G}(V\setminus Attr_\sigma(\{v_i\}, G))$.
    \end{itemize}
    
    With this strategy, there are two possible outcomes. If the token stays in some $V\setminus Attr_\sigma(\{v_i\}, G)$ forever then player $\sigma$ wins. \ Otherwise, for all $i=1,2,\ldots,n$, the token visits $v_i$ infinitely often and player $\sigma$ wins as $V\in F_{\sigma}$. 
\end{proof}

\begin{lemma}\label{L:P1 can't win 0-traps}
Let $\mathcal S=\{S_1,S_2,\ldots, S_k\}\subseteq 2^{V}\setminus \{V\}$ be the collection of  
 all 0-traps in $\mathcal{G}$ and  $V\in \Omega$. \  If for all $S_i\in \mathcal S$, player 1 can't win $\mathcal{G}(S_i)$ then player 0 wins $\mathcal{G}$.
\end{lemma}

\begin{proof}
  Topologically order $\mathcal S$: $S_1 <\ldots< S_k$. By assumption, 
    player 0 wins $\mathcal{G}(S_{1})$. Consider the game $\mathcal{G}(S_\ell)$, $1<\ell  \leq k$. Then $Win_1(\mathcal{G}(S_\ell))\subset S_\ell$ since player 1 can't win $\mathcal{G}(S_\ell)$. If $Win_1(\mathcal{G}(S_\ell))=\emptyset$ then player 0 wins $\mathcal{G}(S_\ell)$.  Otherwise, $Win_1(\mathcal{G}(S_\ell))$ is a 0-trap in 
   $\mathcal{G}$.  By hypothesis,  player 0 wins $\mathcal{G}(Win_1(\mathcal{G}(S_\ell)))$. This is a contradiction. Thus, for all $S_i\in \mathcal S$, player 0 wins $\mathcal{G}(S_i)$. Then for all $v\in V$, $Attr_0(\{v\}, G)=V$ or player 0 wins $\mathcal{G}(V\setminus Attr_0(\{v\}, G))$ since $V\setminus Attr_0(\{v\}, G)$ is a 0-trap in $\mathcal{G}$. By Lemma \ref{L:attract or wins}, player 0 wins $\mathcal{G}$.
\end{proof}

The next lemma shows that we can reduce the size of the wining condition set $\Omega$ if one of the sets $W\in \Omega$ 
is minimal (with respect to $\subseteq$) and not forced-connected.

\begin{lemma}\label{L:unforce-connected useless}
Let $W\subseteq V$ be a subgame. If $\mathcal G(W)$ isn't forced-connected  and no winning set in $\Omega$ is contained in $W$, then 
$Win_1(\mathcal{G})=Win_1(\mathcal{G}')$, where $\mathcal{G}'$ is the same as $\mathcal G$ but has the additional winning set: $\Omega'=\Omega\cup \{W\}$.
\end{lemma}
\begin{proof}
It is clear that $Win_1(\mathcal{G}') \subseteq Win_1(\mathcal{G})$. Also, note that the set $Win_1(\mathcal{G})$ is a 0-trap in $\mathcal{G}'$.  For the set $Win_1(\mathcal{G})$ we have two cases.

\smallskip
\noindent
{\em Case 1}: $W\not\subseteq Win_1(\mathcal{G})$. Then 
$\mathcal{G}'(Win_1(\mathcal{G}))=\mathcal{G}(Win_1(\mathcal{G}))$. Hence,  $Win_1(\mathcal{G})=Win_1(\mathcal{G}')$. 

\smallskip
\noindent
{\em Case 2}: $W\subseteq Win_1(\mathcal{G})$. Define set $\mathcal T=\{T_1,T_2,\ldots,T_r\}$ of all subgames such that each $T_i\subseteq Win_1(\mathcal{G})$ 
 and player 1 wins $\mathcal{G}(T_i)$ for $i=1\ldots, r$.   In particular, both 
 $W$ and $Win_1(\mathcal{G})$ belong to $\mathcal T$. Note that the sets $T_i$ do not have to be a $0$-traps. Topologically order $\mathcal T$, say:    $T_1<  \ldots <T_r$.
 
  Inductively on  $\ell=1,2,\ldots,r$, we prove that  player 1 wins each $\mathcal{G}'(T_\ell)$.   Consider $T_{\ell}$. By induction, for all 1-traps $W'\subset T_\ell$ in $\mathcal G(T_\ell)$, player 1 wins $\mathcal{G}(W')$ and by hypothesis player 1 wins $\mathcal{G}'(W')$. Note that for $\ell=1$ this hypothesis is vacuous. 
    \begin{itemize}
        \item If $T_\ell\notin \Omega'$ then for all $v\in T_\ell$, we have either $Attr_1(\{v\},G(T_\ell))=T_\ell$ or player 1 wins the game $\mathcal{G}'(T_\ell\setminus Attr_1(\{v\},G(T_\ell)))$.  Since $T_\ell\setminus Attr_1(\{v\},G(T_\ell))$ is a 1-trap in $\mathcal{G}'(T_\ell)$. By Lemma \ref{L:attract or wins}, player 1 wins $\mathcal{G}'(T_\ell)$.
        \item If $T_\ell=W$ then $T_\ell$ isn't forced-connected and there exists a 0-trap $W'\subset T_\ell$ in the arena $G(T_{\ell})$. We construct a winning strategy for player 1 as follows: If the token is in $W'$ then player 1 forces the token in $W'$ forever, otherwise moves arbitrarily. With this strategy, for any play $\rho$ in $\mathcal G'(T_{\ell})$ we have 
        $\mathsf{Inf}(\rho)\ne T_\ell$. Since $\Omega'\cap 2^{T_\ell}$ equals $\{T_\ell\}$ by the assumption of the lemma, 
        player 1 wins $\mathcal{G}'(T_\ell)$.
        \item If $T_\ell\in \Omega'\setminus \{W\}$ then by Lemma \ref{L:P1 can't win 0-traps}, there exists a 0-trap $W'\subset T_\ell$ such that player 1 wins $\mathcal{G}(W')$. By hypothesis, player 1 wins $\mathcal{G}'(W')$. Then $Attr_1(W', G(T_\ell))=T_\ell$ or player 1 wins $\mathcal{G}'(T_\ell\setminus Attr_1(W', G(T_\ell)))$ since $T_\ell\setminus Attr_1(W', G(T_\ell))$ is a 1-trap in $\mathcal{G}'(T_\ell)$. We construct a winning strategy for player 1 as follows. 
        \begin{itemize}
            \item If the token is in $W'$ then player 1 forces the token in $W'$ forever and follows the winning strategy in $\mathcal{G}'(W')$.
            \item If the token is in $Attr_1(W', G(T_\ell))$ then player 1 forces the token to $W'$.
            \item Otherwise, player 1 follows the winning strategy in $\mathcal{G}'(T_\ell\setminus Attr_1(W', G(T_\ell)))$.
        \end{itemize}
        Thus, any play consistent with the strategy described, will eventually stay either in  $W'$ or $T_\ell\setminus Attr_1(W', G(T_\ell))$ forever and player 1 wins $\mathcal{G}'(T_\ell)$.
    \end{itemize}
    Thus, for all $T_i\in \mathcal T$, player 1 wins $\mathcal{G}'(T_i)$. Since $Win_1(\mathcal{G})$ belongs to $\mathcal T$, player 1 wins $\mathcal{G}'(Win_1(\mathcal{G}))$. We conclude that $Win_1(\mathcal{G})=Win_1(\mathcal{G}')$.
\end{proof}

    Let $\mathcal{G}$ be M\"uller game with $\Omega=\{W_1,W_2,\ldots,W_s\}$. For the next two lemmas and the follow-up theorem we assume that there exists a $W\in \Omega$ such that $\mathcal G(W)$ is forced-connected and $W$ isn't a 1-trap. 
    
    \begin{definition}[Horn's construction]
    The  game  $\mathcal{G}_{W}=(G_{W}, \Omega_{W})$   determined by $W$ is defined as follows:
    
    \smallskip
    \noindent
    1. $V_{W}=V_0 \cup V_1\cup \{\mathbf{W}\}$, where $\mathbf{W}$ is a player 1's new vertex.\\
    2.  $E_{W}=E\cup (V_0\cap W)\times \{\mathbf{W}\} \ \cup  \ \{\mathbf{W}\}\times (E(V_1\cap W)\setminus W)$.        \\
    3.  $\Omega_{W}=(\Omega \cup \{W'\cup \{\mathbf{W}\}\mid W'\in R\})\setminus (R\cup \{W\})$, 
    where the set $R$ is the following  $R=\{W' \mid W'\in \Omega \text{ and } W\subset W'\}$. 

\end{definition}

\noindent
Note that $|\Omega_W|+1=|\Omega|$, and $\mathcal G_W(W)$ is forced-connected. Thus, similar to the lemma above, Horn's construction also reduces the size of $\Omega$. Now our goal is to show that Horn's construction preserves the winners of the original game. This is shown in the next two lemmas. 

\begin{lemma}\label{L:winning if winning reduction game}
    We have $Win_0(\mathcal{G}_{W})\setminus \{\mathbf{W}\} \subseteq Win_0(\mathcal{G})$.
\end{lemma}

\begin{proof}
    Let $\sigma_W$ be a winning strategy for player 0 in game $\mathcal{G}_{W}$ starting at $s\in V$.  
    We now describe a winning strategy for player 0 in $\mathcal{G}$ starting from $s$.  Player 0 plays the game $\mathcal{G}$ by simulating plays $\rho$ consistent with $\sigma_{W}$ in $\mathcal{G}_{W}$. If a play $\rho$ stays out of 
    $\mathbf{W}$, then the player 0 copies $\rho$ in $\mathcal{G}$. Once $\rho$ moves 
    to $\mathbf{W}$, then player 0 in $\mathcal{G}$ moves to any node in $W\cap V_1$. Then player 0 stays in $W$ and uses its strategy to visit every node in $W$.
   If player 1 moves out of $W$ to a node $u$ in $\mathcal{G}$, this will correspond to a move by player 1 
   from $\mathbf{W}$ to $u$ in $\mathcal{G}_{W}$. Player 0 continues on simulating $\rho$. 
   
   Let $\rho'$ be the play in $\mathcal{G}$ consistent with the strategy.  If $\rho$ meets $\mathbf{W}$  finitely often then $\mathsf{Inf}(\rho)=\mathsf{Inf}(\rho')$ and $\mathsf{Inf}(\rho') \in \Omega$. If $\rho$ never moves out of $\mathbf{W}$ from some point on, then 
 $\mathsf{Inf}(\rho')=W$. In both cases player 0 wins. 
 If the simulation leaves $\mathbf{W}$ infinitely often, then $\mathsf{Inf}(\rho)\in \{W'\cup \{\mathbf{W}\}\mid W'\in R\}$ and $W\subseteq \mathsf{Inf}(\rho)$. Therefore 
 $$
 \mathsf{Inf}(\rho')\subseteq \mathsf{Inf}(\rho)\setminus\{\mathbf{W}\}\cup W=\mathsf{Inf}(\rho)\setminus \{\mathbf{W}\}\subseteq \mathsf{Inf}(\rho'),
 $$ 
 and hence $\mathsf{Inf}(\rho')=\mathsf{Inf}(\rho)\setminus \{\mathbf{W}\}\in R$,  and player 0 wins. 
\end{proof}
\noindent
The next lemma is more involved.
\begin{lemma}\label{L:losing if losing reduction game}
    We have $Win_1(\mathcal{G}_{W})\setminus \{\mathbf{W}\} \subseteq Win_1(\mathcal{G})$. 
\end{lemma}

\begin{proof}
Assume that $W'\subseteq V$ determines a subgame in $\mathcal G$. Obviously, $W'$ also determines a  subgame of  $\mathcal G_W$. We call $W'$ {\em extendible} if $W' \cup \{\bf W\}$ is a subgame of $\mathcal G_W$. Note that there could exist unextendible $W'$.  Based on this, we define the following two sets of subgames of the game $\mathcal G$.
The first set $\mathcal A$ is the following set of subgames of $\mathcal G$:
$$
\{ W' \mid \mbox{$W'$ is extendible $\&$ player 1 wins $\mathcal{G}_{W}(W'\cup \{\mathbf{W}\})$}\}.
$$
Note that if $W'\in \mathcal A$ then player 1 wins the subgame $\mathcal G_W(W')$.  The second set $\mathcal B$ is
the following set of subgames of $\mathcal G$:
$$
\{ W' \mid \mbox{$W\not \subseteq W'$ and player 1 wins $\mathcal{G}_{W}(W')$}\}.
$$
Now we define the set $\mathcal S = \mathcal A\cup \mathcal B$. 
We note that  the set $W$ does not belong to $\mathcal S$ because $W\cup \{\mathbf{W}\}$ is not a subgame in $\mathcal{G}_{W}$ and $W \not \in \mathcal B$ by definition of $\mathcal B$.  
To prove the lemma it suffices to show that player 1 wins $\mathcal G(S)$ for all  $S\in \mathcal S$.   

    
    \smallskip
    
    Topologically order $\mathcal S$:   $S_1< S_2 < \ldots <S_s$.  For each $\ell=1,2,\ldots, s$, we want to show that player 1 wins $\mathcal G(S_\ell)$. 
    As player 1 wins $\mathcal{G}_{W}(S_\ell)$, for all 1-traps $S'\subset S_\ell$  player 1 wins $\mathcal{G}_W(S')$. Let $\mathcal T=\{T_1,T_2,\ldots T_t\}\subseteq 2^{S_\ell}\setminus \{S_\ell\}$ be all 1-traps in the game $\mathcal{G}(S_\ell)$. For each $T_i\in \mathcal T$ we reason as follows.

\smallskip
\noindent
{\em Case 1}: $W\subseteq T_i$.  Then $E_{W}(\mathbf{W})\cap T_i=(E_{W}(V_1\cap W)\setminus W)\cap T_i=(E_{W}(V_1\cap W)\setminus W)\cap S_\ell=E_{W}(\mathbf{W})\cap S_\ell$. Since $W\subseteq S_\ell$ implies player 1 wins $\mathcal{G}_W(S_\ell \cup \{\mathbf W\})$ and $E_{W}(\mathbf{W})\cap S_\ell \ne \emptyset$, $T_i\cup \{\mathbf{W}\}$ is also a 1-trap in the game $\mathcal{G}_{W}(S_\ell\cup \{\mathbf W\})$ and player 1 wins $\mathcal{G}_{W}(T_i\cup \{\mathbf{W}\})$. Hence $T_i$ belongs to $\mathcal A$.
 
 \smallskip
\noindent
{\em Case 2}: $W\not \subseteq T_i$. Note that $T_i$ is also a 1-trap in the game $\mathcal{G}_{W}(S_\ell)$ and player 1 wins $\mathcal{G}_{W}(T_i)$. Hence $T_i$ belongs to $\mathcal B$.
   
    \smallskip
\noindent
Thus,  $\mathcal T\subset \mathcal S$ and by hypothesis, player 1 wins all $\mathcal{G}(T_i)$.

    If $S_\ell \in \mathcal B$ then player 1 wins $\mathcal{G}(S_\ell)=\mathcal{G}_{W}(S_\ell)$. Otherwise $S_\ell \in \mathcal A$ and player 1 wins $\mathcal{G}_W(S_\ell\cup \{\mathbf{W}\})$. 
    
    \begin{itemize}
        \item If $S_\ell\notin \Omega$ then for all $v\in S_\ell$, $Attr_1(\{v\}, G(S_\ell))=S_\ell$ or player 1 wins $\mathcal{G}(S_\ell\setminus Attr_1(\{v\}, G(S_\ell)))$ since $S_\ell\setminus Attr_1(\{v\}, G(S_\ell))$ is a 1-trap in the game $\mathcal{G}(S_\ell)$. By Lemma \ref{L:attract or wins}, player 1 wins $\mathcal{G}(S_\ell)$.
        \item Otherwise by Lemma \ref{L:P1 can't win 0-traps}, there is a 0-trap $Q\subset S_\ell\cup \{\mathbf{W}\}$ in  $\mathcal{G}_W(S_\ell\cup \{\mathbf{W}\})$ such that player 1 wins $\mathcal{G}_W(Q)$. 
        \begin{itemize}
            \item  If $\mathbf{W}\notin Q$ then $W\cap V_0\cap Q=\emptyset$ and $Q$ also determines a 0-trap in the game $\mathcal{G}(S_\ell)$. Since $W\not\subseteq Q$, player 1 wins $\mathcal{G}(Q)=\mathcal{G}_{W}(Q)$ and let $Y=Q$.
            \item If $\mathbf{W} \in Q$ then let $Y=Q\setminus \{\mathbf{W}\}$. Note that for all $v\in V_0\cap W\cap Q$, $E_W(v)\cap Q=E_W(v)\cap (S_\ell\cup \{\mathbf{W}\})$ and $|E_W(v)\cap (S_\ell\cup \{\mathbf{W}\})|>1$. Hence $Y$ determines a 0-trap in the game $\mathcal{G}(S_\ell)$. Since player 1 wins $\mathcal{G}_{W}(Y\cup \{\mathbf{W}\})$, $Y$ belongs to $\mathcal A$ and by hypothesis player 1 wins $\mathcal{G}(Y)$.
        \end{itemize}
        Therefore there exists a 0-trap $Y$ in the game $\mathcal{G}(S_\ell)$ such that player 1 wins $\mathcal{G}(Y)$. Also $Attr_1(Y, G(S_\ell))= S_\ell$ or player 1 wins $\mathcal{G}(S_\ell\setminus Attr_1(Y, G(S_\ell)))$ since $S_\ell\setminus Attr_1(Y, G(S_\ell))$ is a 1-trap in the game $\mathcal{G}(S_\ell)$. Then we construct a winning strategy for player 1 in the game $\mathcal{G}(S_\ell)$ as follows.
        \begin{itemize}
            \item If the token is in $Y$ then player 1 forces the token in $Y$ forever and follows the winning strategy in $\mathcal{G}(Y)$.
            \item If the token is in $Attr_1(Y, G(S_\ell))$ then player 1 forces the token to $Y$.
            \item Otherwise, player 1 follows the winning strategy in $\mathcal{G}(S_\ell\setminus Attr_1(Y, G(S_\ell)))$.
        \end{itemize}
    \end{itemize}
    By hypothesis, for all $S_i\in \mathcal S$, player 1 wins $\mathcal{G}(S_i)$. Since player 1 wins $\mathcal{G}_W(Win_1(\mathcal{G}_W))$, $Win_1(\mathcal{G}_W)$ is a 0-trap in $\mathcal{G}_W$ and it's easy to see that $Win_1(\mathcal{G}_W)\setminus \{\mathbf{W}\}$ is also a 0-trap in $\mathcal{G}$. Then if $\mathbf{W}\in Win_1(\mathcal{G}_W)$ then $Win_1(\mathcal{G}_W)\setminus \{\mathbf{W}\}\in \mathcal A$, otherwise $Win_1(\mathcal{G}_W)\setminus \{\mathbf{W}\}\in \mathcal B$. Since $Win_1(\mathcal{G}_W)\setminus \{\mathbf{W}\}\in \mathcal S$, player 1 wins $\mathcal{G}(Win_1(\mathcal{G}_W)\setminus \{\mathbf{W}\})$ and $Win_1(\mathcal{G}_W)\setminus \{\mathbf{W}\}\subseteq Win_1(\mathcal{G})$.    
\end{proof}

\noindent
By Lemmas \ref{L:winning if winning reduction game} and \ref{L:losing if losing reduction game}, we have the following theorem.

\begin{theorem}\label{Thm:reduction game}
   We have $Win_0(\mathcal{G})=Win_0(\mathcal{G}_W)\setminus \{\mathbf{W}\}$ and $Win_1(\mathcal{G})=Win_1(\mathcal{G}_W)\setminus \{\mathbf{W}\}$. \qed
\end{theorem}

We briefly explain the algorithm, presented in Figure \ref{F:Explicit Muller game}, that takes as input an explicit M\"uller game $\mathcal{G}$ and returns the winning regions of the players. Initially, the algorithm orders $\Omega$ topologically: $W_1<W_2<\ldots <W_s$. 
At each iteration, the algorithm modifies the arena and the winning conditions: 
\begin{itemize}
    \item If $W'_i$ doesn't determine a subgame in game $\mathcal G'$ or $\mathcal G'(W'_i)$ isn't forced-connected, $W'_i$ is removed from $\Omega'$.
    \item Otherwise, $\mathcal G'(W'_i)$ is forced-connected, then:
    \begin{itemize}
        \item If $W'_i$ is a 1-trap in $\mathcal G'$ then $Attr_0(W'_i, G')$ is removed from $\mathcal G'$ and added to the winning region of player 0. Note that all $W'\in \Omega'$ with $W'\cap Attr_0(W'_i, G')\ne\emptyset$ are removed.
        \item Otherwise, apply Horn's construction to $\mathcal{G'}$ by setting $\mathcal G'=\mathcal G'_{W'_i}$. In this construction, a new player 1's node $\mathbf{W}'_i$ is added to $G'$, $\mathbf{W}'_i$ is added to all supersets of $W'_i$ in $\Omega'$ and $W'_i$ itself is removed from $\Omega'$, which maintains the topological order of $\Omega'$.
    \end{itemize}
\end{itemize}

\begin{figure}[H]
    \centering 
    \scriptsize
    \begin{tabular}{l}
        \textbf{Input}: An explicit M\"uller game $\mathcal{G}=(G,\Omega)$\\
        \textbf{Output}: The winning regions of player 0 and player 1\\
        topologically order $\Omega$;\\
        $G'=(V'_0,V'_1,E')\leftarrow G=(V_0,V_1,E)$;\\
        $\Omega'\leftarrow \Omega$;\\
        $Win_0\leftarrow \emptyset$;\\
        \textbf{while} $\Omega'\ne \emptyset$ \textbf{do}\\
        \hspace*{4mm}$W'_i\leftarrow \text{pop}(\Omega')$\\
        \hspace*{4mm}\textbf{if} $\mathcal G'(W'_i)$ is forced-connected \textbf{then}\\
        \hspace*{8mm}\textbf{if} $W'_i$ is a 1-trap in $\mathcal G'$ \textbf{then}\\
        \hspace*{12mm}remove $Attr_0(W'_i,G')$ from $\mathcal G'$ and add it to $Win_0$;\\
        \hspace*{8mm}\textbf{else}\\
        \hspace*{12mm}$\mathcal G' \leftarrow \mathcal G'_{W'_i}$;\\
        \hspace*{8mm}\textbf{end}\\
        \hspace*{4mm}\textbf{end}\\
        \textbf{end}\\
        \textbf{return} $Win_0\cap V$ and  $V\setminus Win_0$
    \end{tabular}
    \caption{Algorithm for explicit M\"uller games}
    \label{F:Explicit Muller game}
\end{figure}

\begin{lemma}\label{L:Explicit Muller game each step}
    At the end of each iteration, we have
    $$Win_0(\mathcal{G})=(Win_0(\mathcal{G}')\cup Win_0)\cap V$$ 
    and 
    $$Win_1(\mathcal{G})=Win_1(\mathcal{G}')\cap V.
    $$
\end{lemma}
\begin{proof}
Initially, $\mathcal{G}'=\mathcal{G}$ and $Win_0=\emptyset$, which holds the lemma. Then for $i=1,2,\ldots,s$, we want to show that at the end of $i$th iteration, $Win_0(\mathcal{G})=(Win_0(\mathcal{G}')\cup Win_0)\cap V$ and $Win_1(\mathcal{G})=Win_1(\mathcal{G}')\cap V$. Let $\mathcal{G}''$ be $\mathcal{G}'$ and $Win_0'$ be $Win_0$ at the beginning of $i$th iteration. Let $\mathcal{G}'''$ be $\mathcal{G}'$ and $Win_0''$ be $Win_0$ at the end of $i$th iteration. By hypothesis, $Win_0(\mathcal{G})=(Win_0(\mathcal{G}'')\cup Win_0')\cap V$ and $Win_1(\mathcal{G})=Win_1(\mathcal{G}'')\cap V$. If $W'_i$ doesn't determine a subgame in game $\mathcal{G}'$ or $W'_i$ isn't forced-connected then by Lemma \ref{L:unforce-connected useless}, $W'_i$ can be removed without affecting the winning regions of the players of the game. Otherwise, if $W'_i$ is a 1-trap in $\mathcal{G}'$ then player 0 wins $\mathcal{G}''(Attr_0(W'_i, G''))$ by forcing the token to $W'_i$ and then to go through $W'_i$. Since $Attr_0(W'_i,G'')$ is a 1-trap in $G''$, $Attr_0(W'_i,G'')\subseteq Win_0(\mathcal{G}'')$. Since $\mathcal{G}'''=\mathcal{G}''(V''\setminus Attr_0(W'_i,G''))$, $Win_0(\mathcal{G}'')=Win_0(\mathcal{G}''')\cup Attr_0(W'_i,G'')$ and $Win_1(\mathcal{G}'')=Win_1(\mathcal{G}''')$. Therefore, $Win_0(\mathcal{G})=(Win_0(\mathcal{G}''')\cup Attr_0(W'_i,G'')\cup Win_0')\cap V=(Win_0(\mathcal{G}''')\cup Win_0'')\cap V$ and $Win_1(\mathcal{G})=Win_1(\mathcal{G}''')\cap V$. If $W'_i$ isn't a 1-trap in $\mathcal{G}'$ then by Theorem \ref{Thm:reduction game}, $Win_0(\mathcal{G}'')=Win_0(\mathcal{G}''')\setminus \{\mathbf{W}'_i\}$ and $Win_1(\mathcal{G}'')=Win_1(\mathcal{G}''')\setminus \{\mathbf{W}'_i\}$. Therefore, $Win_0(\mathcal{G})=(Win_0(\mathcal{G}''')\cup Win_0'')\cap V$ and $Win_1(\mathcal{G})=Win_1(\mathcal{G}''')\cap V$. By hypothesis, at the end of each iteration, $\mathcal{G}'$ and $Win_0$ hold the lemma.
\end{proof}


By Lemma \ref{L:Explicit Muller game each step}, we have the following theorem.

\begin{theorem}\label{Thm:Explicit Muller game}
    At the end of the algorithm, we have 
    $$
    Win_0(\mathcal{G})=Win_0\cap V \ \mbox{and} \  
    Win_1(\mathcal{G})=V\setminus Win_0. \qed
    $$
\end{theorem}

At each iteration, at most one player 1's vertex is added and at most $|V'_0|$ edges are added. Therefore, $|V'_0|=|V_0|$, $|V'_1|$ is bounded by $|V_1|+|\Omega|$ and $|E'|$ is bounded by $|E|+|V_0||\Omega|$. For time complexity of the algorithm, there are at most $|\Omega|$ iterations in a run and the most time-consuming operation is to determine if $\mathcal{G}'(W'_i)$ is forced-connected. By Theorem \ref{Thm:|V_0|} and Theorem \ref{Thm:|V_1|}, we have the following theorems.

\begin{theorem}\label{Thm:Explicit Muller game result 1}
    The explicit M\"uller game $\mathcal{G}$ can be solved in time $\mathbf{O}( |\Omega|\cdot( (\sqrt{|V_1|+|\Omega|}+1) (|E|+|V_0||\Omega|)+(|V_1|+|\Omega|)^2 ) )$.  \qed
\end{theorem}

\begin{theorem}\label{Thm:Explicit Muller game result 2}
    The explicit M\"uller game $\mathcal{G}$ can be solved in time $\mathbf{O}( |\Omega|\cdot(|V_0|+|V_1|+|\Omega|)\cdot |V_0|\log |V_0|  )$.  \qed
\end{theorem}

Both of these algorithms beat the bound of Horn's algorithm. Importantly, Theorem \ref{Thm:Explicit Muller game result 2} decreases the degree of $|\Omega|$ from $|\Omega|^3$ in Horn's algorithm to $|\Omega|^2$. Since $|\Omega|$ is bounded by $2^{|V|}$, the improvement is significant.

\subsection{Horn's approach}\label{SS:Horns failure}
    In \cite{horn2008explicit}, Horn considers sensible sets. A winning set $W\in \Omega$ is sensible if it determines a subgame. Initially, all non-sensible sets are removed (this is fine). Then Horn's assumption is that the iteration 
    process (presented in figure \ref{F:Explicit Muller game}) preserves sensibility. So, Horn's analysis doesn't take non-sensible sets into account. When $\mathcal{G}_W$ is built, a winning condition $W'$ that contains $W$ might become non-sensible.  Below provide an example. In Horn's defence, assume we remove all non-sensible winning sets in the current game $\mathcal{G}_W$. Then one needs to prove that this is a correct action. So, it is an intricate interplay between sensibility and maintenance of the winning sets at each iteration. Horn does not address this. Neglecting non-sensible sets makes the proofs of Lemmas 6 and 7 (in \cite{horn2008explicit}) incorrect. Here is an example.

    Let $\mathcal{G}=(G,\Omega)$ be a game where $G$ is shown as the solid graph in figure \ref{F:horns countercase} and $\Omega=\{W_1=\{v_1,u_1\},W_2=\{v_1,u_1,u_2\},W_3=\{v_1,v_2,u_1,u_2\}\}$. During the algorithm, player 0 wins the subgame determined by $W_1$.
    The set $W_1$ isn't a 1-trap.  So the 
    new player 1's node $\mathbf{W}_1$ is added. The sensible set $W_2$ now becomes non-sensible in $G_{W_1}$. Let us assume, again in Horn's defense, that $W_2\cup \{\mathbf{W}_1\}$ is removed. Then player 0 wins the subgame determined by $W_3\cup \{\mathbf{W}_1\}$. 
    Horn's Lemma 7 applied to the game in $G_{W_1}$ that occurs on $W_3\cup \{\mathbf{W}_1\}$ in the new game states that player 1 wins the original game played on $W_3$, where $W_3$ is removed from the original $\Omega$.

    The proof of Lemma 7, in this case,  considers 
    the maximal winning set $W_3$ and assumes that player 1 has a winning strategy on $W_3$  and uses it to contract the desired winning strategy. This is a self-loop argument. To save the proof, assume that $W_2$ is the maximal set
    by Horn. Since $W_2\cup \{\mathbf{W}_1\}$ 
    is non-sensible in $G_{W_1}$, it's removed during the algorithm. However, player 0 wins $\mathcal{G}(W_1)$  and $W_1$ is a 1-trap in $\mathcal{G}(W_2)$. As a result, player 1 has no winning strategy in $\mathcal{G}(W_2)$ and Horn's proof fails. Since Horn reuses the proof of Lemma 7 in the proof of Lemma 6, Horn fails on the proofs of Lemmas 6 and 7.
\begin{figure}[H]
    \centering 
    \includegraphics[width=0.2\textwidth]{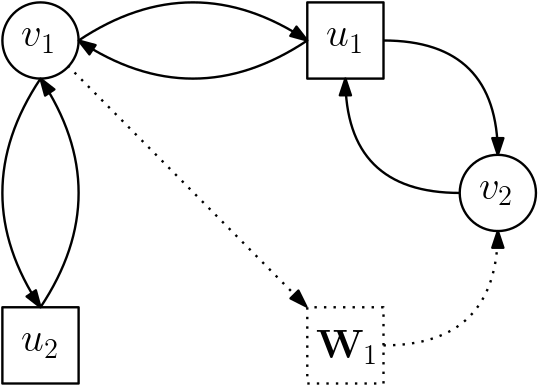} 
    \caption{The counter case of Lemma 6,7 in Horn's paper}
    \label{F:horns countercase}
\end{figure}

    We do not see why Horn's arguments are correct.
    The arguments in Lemmas 7 and 6 have a serious flaw in induction step. 
    In section \ref{SS:proof of Horn}, we develop a new method, totally independent on Horn's considerations. By applying extendible sets and unextendible sets (see the proof of Lemma \ref{L:losing if losing reduction game}), we take the winning sets (which becomes non-sensible) into account, showing that the winning sets stay invariant with each iteration.

\subsection{Applications}

We now apply the results above to specific classes of games. Here we give three examples of such classes. The first such class is the class of fully separated M\"uller games.
A M\"uller game $\mathcal  G$ is {\em fully separated} if for each $W\in \Omega$ there is a $s_W$, called separator, such that for all $s_W\in W$ but $s_W\notin W'$ for all $W'\in \Omega$ distinct from $W$. The second class of games is the class of linear games.  A M\"uller game $\mathcal G$ is a {\em linear game} if the set $\Omega$ forms a linear order $W_1\subset W_2 \subset\ldots \subset W_s$. These classes of games were studied in \cite{ishihara2002complexity}.  As the games are fully separated, when one constructs $\mathcal{G}'_{W'_i}$ there is no need to add a new vertex.  Then applying Theorems \ref{Thm:|V_0|} and  \ref{Thm:|V_1|} to Horn's algorithm, we get the following result:

\begin{theorem}
Each of the following is true: 
\begin{enumerate}
\item 
   Any  fully separated M\"uller game $\mathcal{G}$ can be solved in time $\mathbf{O}(|V|\cdot((\sqrt{|V_1|}+1)|E|+|V_1|^2))$. 
\item
    Any fully separated M\"uller game $\mathcal{G}$ can be solved in time $\mathbf{O}(|V|^2\cdot|V_0|\log |V_0|)$. \qed
\end{enumerate}
\end{theorem}

Both of these algorithms beat the bound of \cite{ishihara2002complexity} $\mathbf{O}(|V|^2 |E|)$ that solves fully separated M\"uller game. Applying Theorem \ref{Thm:Explicit Muller game result 1} and Theorem \ref{Thm:Explicit Muller game result 2}, we have the following theorems.

\begin{theorem}
Each of the following is true:
\begin{enumerate}
\item Any linear M\"uller game \  $\mathcal{G}$  \ can be \  solved \  in time  $\mathbf{O}(|V|\cdot((\sqrt{|V|}+1)\cdot |V_0||V|+|V|^2))$. 
\item  Any linear M\"uller game $\mathcal{G}$ \ can be \ solved \ in time $\mathbf{O}(|V|^2\cdot|V_0|\log |V_0|)$. \qed
\end{enumerate}

\end{theorem}

Both of these algorithms beat the bound  $\mathbf{O}(|V|^{2\cdot |V|-1}|E|)$ from of \cite{ishihara2002complexity} and the bound  $\mathbf{O}(|V|^3\cdot |V_0|)$ implied from Horn's algorithm. 

The third class of M\"uller games was introduced by A. Dawar and P. Hunter in \cite{hunter2008complexity}. They investigated games with anti-chain winning condition. A winning condition $\Omega$ is an {\em anti-chain} if  $X\not \subseteq Y$ for all $X,Y\in \Omega$. Applying Theorem \ref{Thm:Explicit Muller game result 1} and Theorem \ref{Thm:Explicit Muller game result 2}, we have the following theorems. Note that, since the winning condition is an anti-chain, $|V'_1|$ is bounded by $|V_1|$, $|E'|$ is bounded by $|E|$ and no new player 1's vertex is added to $\Omega'$.

\begin{theorem}
Each of the following is true:
\begin{enumerate}
\item 
Any M\"uller game $\mathcal{G}$ with anti-chain winning condition can be solved in time $\mathbf{O}(|\Omega|\cdot((\sqrt{|V_1|}+1)|E|+|V_1|^2))$.  
\item     Any M\"uller game $\mathcal{G}$ with anti-chain winning condition can be solved in time $\mathbf{O}( |\Omega||V|\cdot |V_0|\log |V_0|)$. \qed
\end{enumerate}
\end{theorem}

Just as above, both of the algorithms beat the bound  $\mathbf{O}(|\Omega||V|^2|E|)$ from \cite{hunter2008complexity} and the bound of Horn's algorithm $\mathbf{O}(|\Omega||V||E|)$ that solves the explicit M\"uller games with anti-chain winning conditions.

\bibliographystyle{plain}
\bibliography{bibfile}


\section{Appendix: The  soft-threshold search  for  the ISCCM problem}\label{S:soft-threshold search algorithm for ISCCM problem}


For completeness of this paper, we compactly present the  soft-threshold search algorithm  for  the ISCCM problem borrowed from \cite{haeupler2012incremental}.  
 We represent the topological order by the dynamic order list of distinct elements so that three basic operations, order queries (does $x$ occur before $y$ in the list?), deletions, and insertions (insert a given non-list element just before, or just after, a given list element) are fast. With any of methods in \cite{bender2002two}, \cite{dietz1987two}, each of the operation on the topological order takes a constant factor time. Note that the initial topological order is given by a sequence of size $n$ and we can turn it into a dynamic order list in $\mathbf{O}(n)$ time.

\begin{figure}[H]
    \centering 
    \scriptsize
    \begin{tabular}{l}
        edge \textbf{function} COMPATIBLE-SEARCH(vertex $v$, vertex $w$)\\
        \hspace*{4mm}$F=\{w\}$; $B=\{v\}$; \\
        \hspace*{4mm}$A_F=\{(w,x)\mid (w,x)\text{ is an edge}\}$; $A_B=\{(y,v)\mid (y,v)\text{ is an edge}\}$\\
        \hspace*{4mm}\textbf{while} $\exists (u,x)\in A_F$, $\exists (y,z)\in A_B (u<z)$ \textbf{do}\\
        \hspace*{8mm}choose $(u,x)\in A_F$ and $(y,z)\in A_B$ with $u<z$\\
        \hspace*{8mm}$A_F=A_F-\{(u,x)\}$; $A_B=A_B-\{(y,z)\}$\\
        \hspace*{8mm}\textbf{if} $x\in B$ \textbf{then return} $(u,x)$ \textbf{else if} $y\in F$ \textbf{then return} $(y,z)$\\
        \hspace*{8mm}\textbf{if} $x\notin F$ \textbf{then}\\
        \hspace*{12mm}$F=F\cup \{x\}$; $A_F=A_F\cup\{(x,q)\mid (x,q)\text{ is an edge}\}$\\
        \hspace*{8mm}\textbf{end}\\
        \hspace*{8mm}\textbf{if} $y\notin B$ \textbf{then}\\
        \hspace*{12mm}$B=B\cup \{y\}$; $A_B=A_B\cup\{(r,y)\mid (r,y)\text{ is an edge}\}$\\
        \hspace*{8mm}\textbf{end}\\
        \hspace*{4mm}\textbf{end}\\
        \hspace*{4mm}\textbf{return} $null$ 
    \end{tabular}
    \caption{Implementation of compatible search \cite{haeupler2012incremental}}
    \label{F:compatible search}
\end{figure}

The compatible search is defined as follows. Since adding an edge $(v,w)$ with $v<w$ doesn't change the topological order, we assume $v>w$.  When a new edge $(v,w)$ is added, concurrently search forward from $w$ and backward from $v$. Each step of the search traverses one edge $(u,x)$ forward and one edge $(y,z)$ backward. Two edges $(u,x)$ and $(y,z)$ are compatible if $u<z$. The compatible search only examines the compatible pairs of edges. During the search, every vertex is in one of three states: unvisited, forward(first visited by the forward search), or backward(first visited by the backward search). The search maintains the set $F$ of forward vertices, the set $B$ of backward vertices, the set $A_F$ of edges to be traversed forward and the set $A_{B}$ of edges to be traversed backward. If the search does not detect a cycle, vertices in $B\cup F$ must be reordered to restore topological order and the search return null. Otherwise the search returns an edge other than $(v,w)$ on the cycle.  The compatible search is done by calling COMPATIBLE-SEARCH$(v,w)$ function presented in Figure \ref{F:compatible search}. 

If the search returns null, restore topological order as follows. Let $t=\min(\{v\}\cup \{u\mid \exists (u,x)\in A_F\})$,  $F_<=\{x\in F\mid x < t\}$ and $B_>=\{y\in B\mid y>t \}$. If $t=v$, move all vertices in $F_<$ just after $t$. Otherwise, move all vertices in $F_<$ just before $t$ and all vertices in $B_>$ just before all vertices in $F_<$.

\begin{figure}[H]
    \centering 
    \scriptsize
\begin{tabular}{l}
    \textbf{macro} SEARCH-STEP(vertex $u$, vertex $z$)\\
    \hspace*{4mm}$(u,x)$=\textit{out}$(u)$; $(y,z)$=\textit{in}$(z)$\\
    \hspace*{4mm}\textit{out}$(u)$=\textit{next-out}$((u,x))$; \textit{in}$(z)$=\textit{next-in}$((y,z))$\\
    \hspace*{4mm}\textbf{if} \textit{out}$(u)$=\textit{null} \textbf{then} $F_A=F_A-\{u\}$; \textbf{if} \textit{in}$(z)$=\textit{null} \textbf{then} $B_A=B_A-\{z\}$\\
    \hspace*{4mm}\textbf{if} $x\in B$ \textbf{then return} $(u,x)$ \textbf{else if} $y\in F$ \textbf{then return} $(y,z)$\\
    \hspace*{4mm}\textbf{if} $x\notin F$ \textbf{then}\\
    \hspace*{8mm}$F=F\cup \{x\}$; \textit{out}$(x)$=\textit{first-out}$(x)$\\
    \hspace*{8mm}\textbf{if} \textit{out}$(x)$$\ne$\textit{null} \textbf{then} $F_A=F_A\cup\{x\}$\\
    \hspace*{4mm}\textbf{end}\\
    \hspace*{4mm}\textbf{if} $y\notin B$ \textbf{then}\\
    \hspace*{8mm}$B=B\cup \{y\}$; \textit{in}$(y)$=\textit{first-in}$(y)$\\
    \hspace*{8mm}\textbf{if} \textit{in}$(y)$$\ne$\textit{null} \textbf{then} $B_A=B_A\cup\{y\}$\\
    \hspace*{4mm}\textbf{end}    
\end{tabular}
    \caption{Implementation of a search step \cite{haeupler2012incremental}}
    \label{F:search step}
\end{figure}

\begin{figure}[H]
    \centering 
    \scriptsize
\begin{tabular}{l}
    edge \textbf{function} SOFT-THRESHOLD-SEARCH(vertex $v$, vertex $w$)\\
    \hspace*{4mm}$F=\{w\}$; $B=\{v\}$; \textit{out}$(w)$=\textit{first-out}$(w)$; \textit{in}$(v)$=\textit{first-in}$(v)$; $s=v$\\
    \hspace*{4mm}\textbf{if} \textit{out}$(w)$=\textit{null} \textbf{then} $F_A=\{\}$ \textbf{else} $F_A=\{w\}$; $F_P=\{\}$\\
    \hspace*{4mm}\textbf{if} \textit{in}$(v)$=\textit{null} \textbf{then} $B_A=\{\}$ \textbf{else} $B_A=\{v\}$; $B_P=\{\}$\\
    \hspace*{4mm}\textbf{while} $F_A\ne \{\}$ and $B_A\ne \{\}$ \textbf{do}\\
    \hspace*{8mm}choose $u\in F_A$ and $z\in B_A$\\
    \hspace*{8mm}\textbf{if} $u<z$ \textbf{then} SEARCH-STEP$(u, z)$ \textbf{else}\\
    \hspace*{12mm}\textbf{if} $u>s$ \textbf{then}\\
    \hspace*{16mm}$F_A=F_A-\{u\}$; $F_P=F_P\cup \{u\}$\\
    \hspace*{12mm}\textbf{end}\\
    \hspace*{12mm}\textbf{if} $z<s$ \textbf{then}\\
    \hspace*{16mm}$B_A=B_A-\{z\}$; $B_P=B_P\cup \{z\}$\\
    \hspace*{12mm}\textbf{end}\\
    \hspace*{8mm}\textbf{end}\\
    \hspace*{8mm}\textbf{if} $F_A=\{\}$ \textbf{then}\\
    \hspace*{12mm}$B_P=\{\}$; $B_A=B_A-\{s\}$\\
    \hspace*{12mm}\textbf{if} $F_P\ne \{\}$ \textbf{then}\\
    \hspace*{16mm}choose $s\in F_P$; $F_A=\{x\in F_P\mid x\le s\}$; $F_P=F_P-F_A$\\
    \hspace*{12mm}\textbf{end}\\
    \hspace*{8mm}\textbf{end}\\
    \hspace*{8mm}\textbf{if} $B_A=\{\}$ \textbf{then}\\
    \hspace*{12mm}$F_P=\{\}$; $F_A=F_A-\{s\}$\\
    \hspace*{12mm}\textbf{if} $B_P\ne \{\}$ \textbf{then}\\
    \hspace*{16mm}choose $s\in B_P$; $B_A=\{x\in B_P\mid x\ge s\}$; $B_P=B_P-B_A$\\
    \hspace*{12mm}\textbf{end}\\
    \hspace*{8mm}\textbf{end}\\
    \hspace*{4mm}\textbf{end}\\
    \hspace*{4mm}\textbf{return} $null$ 
\end{tabular}
    \caption{Implementation of soft-threshold search \cite{haeupler2012incremental}}
    \label{F:soft-threshold search}
\end{figure}

For each vertex, we maintain a list of its outgoing edges and a list of its incoming edges implemented by singly linked lists, which we call the outgoing list and incoming list, respectively. \textit{first-out}$(x)$ and \textit{first-in}$(x)$ are the first edge on the outgoing list and the incoming list of $x$, respectively. \textit{next-out}$((x,y))$ and \textit{next-in}$((x,y))$ are the edges after $(x,y)$ on the outgoing list of $x$ and the incoming list of $y$, respectively. In each case, if there is no such edge, the value is null. Adding a new edge to the graph takes $\mathbf{O}(1)$ time. 

During the search, we do not maintain $A_F$ and $A_B$ explicitly. Instead we partition $F$ and $B$ into active, passive and dead vertices. 
A vertex in $F$ is live if it has at least one outgoing untraversed edge and a vertex in $B$ is live if it has at least one incoming untraversed edge. A vertex in $F\cup B$ is dead if it's not live. The live vertices in $F\cup B$ are partitioned into active vertices and passive vertices. We maintain the sets $F_A$ and $F_P$, and $B_A$ and $B_P$, of active and passive vertices in $F$ and $B$. Vertex $s$ is a vertex in $F\cup B$, initially $v$. All vertices in $F_P$ are greater than $s$ and all vertices in $B_P$ are less than $s$. Vertices in $F_A\cup B_A$ can be on either side of $s$. Searching continues until $F_A=B_A=\emptyset$. During a search, the algorithm chooses $u$ from $F_A$ and $z$ from $B_A$ arbitrarily. If $u<z$ then the algorithm traverses an edge out of $u$ and an edge into $z$ and makes each newly live vertex active. If $u>z$ then the algorithm traverses no edges and some (at least one) of $u$ and $z$ become passive. When $F_A$ or $B_A$ becomes empty, the algorithm updates $s$ and other structures as follows. Assume $F_A$ is empty and the updating is symmetric if $B_A$ is empty. All vertices in $B_P$ become dead and $s$ becomes dead if it is live in $B_A$. A new $s$ is selected to be the median of the choices in $F_P$ and make all vertices $x\in F_A$ with $x\le s$ active.

The soft-threshold search is done by calling SOFT-THRESHOLD-SEARCH$(v, w)$ function presented in Figure \ref{F:soft-threshold search}. It uses an auxiliary macro SEARCH-STEP presented in Figure \ref{F:search step}, intended to be expanded in-line. If SOFT-THRESHOLD-SEARCH$(v, w)$ returns null, let $t=\min(\{v\}\cup \{x\in F\mid \textit{out}(x)\ne \textit{null}\})$ and determine the sets $F_<$ and $B_>$. Topologically sort these acyclic subgraphs induced by $F_<$ and $B_>$ by linear-time methods. Then reorder the vertices in $F_<$ and $B_>$ as discussed in compatible search algorithm.

Then soft-threshold search algorithm is extended to the maintenance of SCCs. For each SCC, maintain a canonical vertex which represents the SCC. We use disjoint set data structures \cite{tarjan1975efficiency,tarjan1984worst} to represent the vertex partition defined by SCCs. This structure supports the query \textit{find}$(v)$ which returns the canonical vertex of the component containing vertex $v$, and the operation \textit{unite}$(x,y)$, which given canonical vertices $x$ and $y$, combines the sets containing $x$ and $y$, and makes $x$ be the canonical vertex of the new set. Using path compression and union by rank techniques, the amortized time per find is $\mathbf{O}(1)$ and the total time charged for the unites is $\mathbf{O}(\min(k+1, n)\log (\min(k+1, n)))$ \cite{tarjan1984worst} since there are at most $\min(k+1,n)$ united vertices over $k$ edge additions.

Then during the soft-threshold search, we only maintain the canonical vertices and their outgoing edges and incoming edges. The incoming and outgoing lists are circular. When an edge addition combines several components into one, we combine the incoming lists and outgoing lists to the incoming list and outgoing list of the new component and delete from the order those that are no longer canonical, which takes $\mathbf{O}(1)$ time per old component. This may create multiple edges between the same pair of SCCs and loops, edges whose ends are in the same SCC. We delete the loops during the search and each deletion of a loop takes $\mathbf{O}(1)$ time.

Since we identify each strong component with its canonical vertex, when an edge $(v,w)$ is added, we examine the topological order of \textit{find}$(v)$ and \textit{find}$(w)$. It \textit{find}$(v)$$>$\textit{find}$(w)$, do a soft-threshold search by calling SOFT-THRESHOLD-SEARCH(\textit{find}$(v)$, \textit{find}$(w)$) function presented in Figure \ref{F:soft-threshold search} but with the macro SEARCH-STEP redefined as in Figure \ref{F:SCC search step}. The new version of SEARCH-STEP is different from the old one by only visiting the canonical vertices, using circular edge lists and not doing cycle detection. This SOFT-THRESHOLD-SEARCH always returns null. 

\begin{figure}[H]
    \centering 
    \scriptsize
\begin{tabular}{l}
    \textbf{macro} SEARCH-STEP(vertex $u$, vertex $z$)\\
    \hspace*{4mm}$(q,g)$=\textit{out}$(u)$; $(h,r)$=\textit{in}$(z)$\\
    \hspace*{4mm}\textit{out}$(u)$=\textit{next-out}$((q,g))$; \textit{in}$(z)$=\textit{next-in}$((h,r))$\\
    \hspace*{4mm}$x=$\textit{find}$(g)$; \textbf{if} \textit{out}$(u)$=\textit{first-out}$(u)$ \textbf{then} $F_A=F_A-\{u\}$\\
    \hspace*{4mm}$y=$\textit{find}$(h)$; \textbf{if} \textit{in}$(z)$=\textit{first-in}$(z)$ \textbf{then} $B_A=B_A-\{z\}$\\
    \hspace*{4mm}\textbf{if} $u=x$ \textbf{then} delete $(q,g)$; \textbf{if} $y=z$ \textbf{then} delete $(h,r)$\\
    \hspace*{4mm}\textbf{if} $x\notin F$ \textbf{then}\\
    \hspace*{8mm}$F=F\cup \{x\}$; \textit{out}$(x)$=\textit{first-out}$(x)$\\
    \hspace*{8mm}\textbf{if} \textit{out}$(x)$$\ne$\textit{null} \textbf{then} $F_A=F_A\cup\{x\}$\\
    \hspace*{4mm}\textbf{end}\\
    \hspace*{4mm}\textbf{if} $y\notin B$ \textbf{then}\\
    \hspace*{8mm}$B=B\cup \{y\}$; \textit{in}$(y)$=\textit{first-in}$(y)$\\
    \hspace*{8mm}\textbf{if} \textit{in}$(y)$$\ne$\textit{null} \textbf{then} $B_A=B_A\cup\{y\}$\\
    \hspace*{4mm}\textbf{end}    
\end{tabular}
    \caption{Redefinition of SEARCH-STEP to find SCCs \cite{haeupler2012incremental}}
    \label{F:SCC search step}
\end{figure}

Once the search finishes, let $t=\min(\{\textit{find}(v)\}\cup \{x\in F\mid \textit{out}(x)\ne \textit{null}\})$. Compute the sets $F_<$ and $B_>$. Detect and find the new SCC by running a linear-time algorithm on the subgraph induced by the SCCs where the vertex set is $X=F_<\cup \{t\} \cup B_>$ and the edge set is $Y=\{(\textit{find}(x), \textit{find}(y))\mid \textit{find}(x)\ne \textit{find}(y) \text{ and } (x,y) \text{ is an edge with } \textit{find}(x)\in F_< \text{ or } \textit{find}(y)\in B_{>}\}$. If a new SCC is found, combine the old SCCs into one new SCC with canonical vertex $v$. Reorder the list of vertices in topological order by moving the vertices in $X-\{t\}$ as discussed in compatible search algorithm and delete all vertices that are no longer canonical from the order.


\end{document}